\documentclass[conference]{IEEEtran}

\usepackage{booktabs} % For formal tables
\usepackage[linesnumbered,ruled,noend]{algorithm2e}
\usepackage[noend]{algorithmic}
\usepackage{listings}
\usepackage{threeparttable}
\usepackage{tikz}
\usepackage[T1]{fontenc}
\usepackage{pgfplots}
\usepackage{filecontents}
\usepackage{amsmath}
\usepackage{amssymb}
\usepackage{txfonts}
\usepackage{subfigure}
\usepackage{multirow}

\newcommand{\nop}[1]{}
\newcommand{\tabincell}[2]{\begin{tabular}{@{}#1@{}}#2\end{tabular}}

\newtheorem{definition}{Definition}
\newtheorem{theorem}{Theorem}

\newtheorem{example}{Example}

\title{Accelerating Partial Evaluation in Distributed SPARQL Query Evaluation}

\author{%
{Peng Peng{\small $~^{1}$}, Lei Zou{\small $~^{2,3,4}$}, Runyu Guan{\small $~^{1}$}}%
\vspace{1.6mm}\\
\fontsize{10}{10}\selectfont\itshape
$^{1}$\,Hunan University, Changsha, China\\
$^{2}$\,Peking University, Beijing, China\\
$^{3}$\,Beijing Institute of Big Data Research, China\\
$^{4}$\,National Engineering Laboratory for Big Data Analysis Technology and Application (PKU), China\\
\\
\fontsize{9}{9}\selectfont\ttfamily\upshape
$^{1}$\{hnu16pp,guanrunyu\}@hnu.edu.cn, $^{2,3,4}$\,zoulei@pku.edu.cn
}

\begin{document}

\maketitle

\begin{abstract}
Partial evaluation has recently been used for processing SPARQL queries over a large resource description framework (RDF) graph in a distributed environment. However, the previous approach is inefficient when dealing with complex queries. In this study, we further improve the ``partial evaluation and assembly'' framework for answering SPARQL queries over a distributed RDF graph, while providing performance guarantees. Our key idea is to explore the intrinsic structural characteristics of partial matches to filter out irrelevant partial results, while providing performance guarantees on a network trace (data shipment) or the computational cost (response time). We also propose an efficient assembly algorithm to utilize the characteristics of partial matches to merge them and form final results. To improve the efficiency of finding partial matches further, we propose an optimization that communicates variables' candidates among sites to avoid redundant computations. In addition, although our approach is partitioning-tolerant, different partitioning strategies result in different performances, and we evaluate different partitioning strategies for our approach. Experiments over both real and synthetic RDF datasets confirm the superiority of our approach.
\end{abstract}

\section{Introduction}\label{sec:Introduction}
The resource description framework (RDF) is a semantic web data model that represents data as a collection of triples of the form $\langle$subject, property, object$\rangle$. 
An RDF dataset can also be represented as a graph where subjects and objects are vertices, and triples are edges with labels between vertices. Meanwhile, SPARQL is a query language designed for retrieving and manipulating an RDF dataset, and its primary building block is the basic graph pattern (BGP). A BGP query can also be seen as a query graph, and answering a BGP query is equivalent to finding subgraph matches of the query graph over the RDF graph. In this study, we focus on the evaluation of BGP queries. An example SPARQL query of four triple patterns (e.g., ?t label ?l) is listed in the following, and retrieves all people influencing Crispin Wright and their interests:

\nop{
An RDF dataset can also be represented
as a graph, where subjects and objects are vertices and triples are
edges with labels between vertices.
On the other hand, SPARQL is a query language designed for retrieving and manipulating an RDF dataset, and its primary building block is the basic graph pattern (BGP). A BGP query can also be seen as a query graph, and answering a BGP Q is equivalent to finding subgraph matches of the query graph over RDF graph. In this paper, we focus on the evaluation of BGP queries. An example SPARQL query of four triple patterns (e.g., ?t label ?l) is listed in the following, which retrieves all people influencing Crispin Wright and their interests.

A triple can be naturally seen as a pair of entities connected by a named relationship or an entity associated with a named attribute value. An RDF dataset can also be represented as a graph where subjects and objects are vertices, and triples are edges with property names as edge labels.
}

\small{
\begin{lstlisting}[language=SQL]
Select ?p2, ?l where {?t label ?l.
?p1 influencedBy ?p2.  ?p2 mainInterest ?t.
?p1 name ``Crispin Wright''@en.}.
\end{lstlisting}
}
\normalsize

With the increasing size of RDF data published on the Web, it is necessary for us to design a distributed database system to process SPARQL queries. In many applications, the RDF graphs are geographically or administratively distributed over the sites, and the RDF repository partitioning strategy is not controlled by the distributed RDF system itself. For example, the European Bioinformatics Institute\footnote{https://www.ebi.ac.uk/rdf/} has built up a uniform platform for users to query multiple bioinformatics RDF datasets, including BioModels, Biosamples, ChEMBL, Ensembl, Atlas, Reactome, and OLS. These datasets are provided by different data publishers, and should be administratively partitioned according to their data publishers. Thus, partitioning-tolerant SPARQL processing is desirable.

\nop{
With the increasing size of RDF data published on the Web, it is necessary for us to design a distributed database system to process SPARQL queries. In many applications, the RDF graph are geographically or administratively distributed over the sites, and the RDF repository partitioning strategy is not controlled by the distributed RDF system itself. For example, European Bioinformatics Institute has built up a uniform platform for users to query multiple bioinfomatics RDF datasets, including BioModels, Biosamples, ChEMBL, Ensembl, Atlas, Reactome and OLS. These datasets are provided by different data publishers and had better be administratively partitioned according to their data publishers. Thus, partitioning-tolerant SPARQL processing is desirable.
}

For partitioning-tolerant SPARQL processing on distributed RDF graphs, Peng et al.\cite{DBLP:journals/corr/PengZOCZ14} discuss how to evaluate SPARQL queries in a ``partial evaluation and assembly'' framework. However, the framework's efficiency has significant potential for improvement. Its major bottleneck is the large volume of partial evaluation results, leading to a high cost for generating and assembling the results.

%discuss how to evaluate SPARQL queries in the ``partial evaluation and assembly'' framework. However, its efficiency has great potential to be improved. Its major bottleneck is the large volume of partial evaluation results, which causes such a high cost for generating and assembling them.

%In this study, we propose several optimizations for the “partial evaluation and assembly” framework
In this study, we propose several optimizations for the ``partial evaluation and assembly'' framework \cite{DBLP:journals/corr/PengZOCZ14}, to prune the irrelevant partial evaluation results, and assemble them efficiently to form the final results. The first step is to compress all partial evaluation results into a compact data structure named the local partial match equivalence class (LEC) feature. Then, we can communicate the LEC features among sites to filter out some irrelevant partial evaluation results. We can prove that the proposed optimization technique is \emph{partition bounded} in both \emph{response time} and \emph{data shipment}  \cite{DBLP:journals/pvldb/FanWWD14}. The second step is to assemble all local partial matches based on their LEC features. Finally, to avoid further redundant computations within the sites, we propose an optimization that communicates variables' candidates among the sites to prune some irrelevant candidates. In addition, although our approach is partitioning-tolerant, different partitioning strategies result in different performances, and we also evaluate different partitioning strategies for our approach.

%to prune the irrelevant partial evaluation results and assemble them efficiently to form the final results. The first is to compress all partial evaluation results into a compact data structure named \emph{LEC feature}. Then, we can communicate the LEC features among sites to filter out some irrelevant partial evaluation results. We can prove that the proposed optimization technique is \emph{partition bounded} in both \emph{response time} and \emph{data shipment} 

\nop{
In addition, although our framework is partitioning-tolerant, the partitioning strategies still impact the query performance, so we need propose a good partitioning strategy special for our framework. In Section \ref{sec:Partitioning}, based on the complexity analysis of our framework, we find that a partitioning strategy where crossing edges are far from each other can reduce communication cost. Then, we propose a partitioning strategy that use the maximal edge independent sets to partition the RDF graph. Our partitioning strategy can also provide a good performance guarantee of our framework.
}

Thus, we make the following contributions in this study.

\begin{itemize}
\item We explore the intrinsic structural characteristics of partial results to compress them into a compact data structure, the \emph{LEC feature}. We communicate and utilize the LEC features to prune some irrelevant partial results. We prove theoretically that the LEC feature can guarantee the performance of the pruning optimization in both \emph{response time} and \emph{data shipment}.
  \item	We propose an efficient LEC feature-based assembly algorithm to merge all the partial results together and form the final results.
  \item We present an optimization based on the communication of the variables' internal candidates among different sites to avoid further redundant computations within the sites.
    \item We define a specific cost model for our method to measure the cost of different partitioning strategies, and to select the best partitioning from the existing partitionings.
  \item We conduct experiments over both real and synthetic RDF datasets to confirm the superiority of our approach.
\end{itemize}

\nop{
The rest of the paper is organized as follows. Section \ref{sec:background} provides the fundamental definitions and introduces the background of  ``partial evaluation and assembly'' framework. Section \ref{sec:background} gives an overview of our approach. Definitions of LEC features is covered in Section \ref{sec:LPMIsomorphism} and the assembly of LEC features to compute the final query result is discussed in Section \ref{sec:Assembly}. We evaluate our approach, both in terms of its internal characteristics and in terms of its relative performance against other approaches in Section \ref{sec:Experiment}. We discuss related work in the areas of distributed SPARQL query processing and partial query evaluation in Section \ref{sec:relatedwork}. Section \ref{sec:Conclusion} concludes the paper and outlines some future research directions.
}

\section{Background}
\label{sec:background}
\subsection{Distributed RDF Graph and SPARQL Query}
\label{sec:RDFandSPARQL}

An RDF dataset can be represented as a graph where subjects and objects are vertices, and triples are labeled edges. In this study, an RDF graph $G$ is vertex-disjoint-partitioned into a number of \emph{fragments}, each of which resides at one site. The vertex-disjoint partitioning methods guarantee that there are no overlapping vertices between fragments. Here, to guarantee data integrity and consistency, we store some replicas of crossing edges. Formally, we define the \emph{distributed RDF graph} as follows:

%An RDF dataset can be represented as a graph where subjects and objects are vertices and triples are labeled edges. In this paper, an RDF graph $G$ is vertex-disjoint partitioned into a number of \emph{fragments}, each of which resides at one site.
%The vertex-disjoint partitioning methods guarantee that there are no overlapping vertices between fragments. Here, to guarantee data integrity and consistency, we store some replicas of crossing edges. Formally, we define the \emph{distributed RDF graph} as follows.
%Since the RDF graph $G$ is partitioned by our system, metadata is readily available regarding crossing edges (both outgoing and incoming edges) and the endpoints of crossing edges.

\begin{definition} \label{def:internaldistributedgraph} \textbf{(Distributed RDF Graph)} Let $u$ and $\overrightarrow {uu^\prime}$ denote the vertex and edge in an RDF graph. A distributed RDF graph $G=\{V,E, \Sigma \}$ consists of a set of fragments $\mathcal{F} = \{F_1,F_2,...,F_k\}$, where each $F_i$ is specified by $(V_i \cup V_i^e, E_i \cup E_i^c, \Sigma_i)$ ($i=1,...,k$) such that:
\begin{enumerate}
\item $\{V_1,...,V_k\}$ is a partitioning of $V$, i.e., $V_i  \cap V_j  = \emptyset ,1 \le i,j \le k,i \ne j $ and $\bigcup\nolimits_{i = 1,...,k} {V_i  = V}$ ;
\item $E_i \subseteq V_i \times V_i$, $i=1,...,k$;

\item $E_i^c$ is a set of crossing edges between $F_i$ and other fragments, i.e.,
\begin{multline*}
 E_i ^c = (\bigcup\nolimits_{1 \le j \le k \wedge j \ne i} {\{ \overrightarrow {uu^\prime} |u \in F_i  \wedge u^\prime \in F_j  \wedge \overrightarrow{uu^\prime}  \in E \} }) \\ \bigcup
 (\bigcup\nolimits_{1 \le j \le k \wedge j \ne i} {\{ \overrightarrow {u^\prime u} |u \in F_i  \wedge u^\prime \in F_j  \wedge \overrightarrow{u^\prime u} \in E   \}} )
 \end{multline*}

 \item A vertex $u^\prime \in V_i^e$ if and only if vertex $u^\prime$ resides in other fragment $F_j$ and $u^{\prime}$ is an endpoint of a crossing edge between fragment $F_i$ and $F_j$ ($F_i \neq F_j$),  i.e.,
\begin{multline*}
 V_i^e = (\bigcup\nolimits_{1 \le j \le k \wedge j \ne i} \{ {u^\prime} |\overrightarrow {uu^\prime}$  $\in E_i ^c \wedge u \in F_i \} ) \bigcup \\
 (\bigcup\nolimits_{1 \le j \le k \wedge j \ne i} {\{ {u^\prime} |\overrightarrow {u^\prime u}  \in E_i ^c \wedge u \in F_i \} })
 \end{multline*}

\item Vertices in $V_i^e$ are called \emph{extended} vertices of $F_i$, and vertices in $V_i$ are called \emph{internal} vertices of $F_i$; and%, and vertices in $V_i$ adjacent to vertices in $V_i^e$ are called \emph{boundary} vertices of $F_i$;
\item $\Sigma_i$ is a set of edge labels in $F_i$.
\end{enumerate}
\end{definition}

\begin{example}\label{example:dataandquery}
Fig. \ref{fig:datagraph} shows a distributed RDF graph $G$ consisting of three fragments $F_1$, $F_2$, and $F_3$. The numbers besides the vertices are vertex IDs that are introduced for ease of presentation. In Fig. \ref{fig:datagraph}, $\overrightarrow{001,006}$ and $\overrightarrow{006,005}$ are crossing edges between $F_1$ and $F_2$. In addition, edge $\overrightarrow{001,012}$ is a crossing edge between $F_1$ and $F_3$. Hence, $V_1^e=\{006,012\}$ and $E_1^c=\{\overrightarrow{001,006},$  $\overrightarrow{006,005},\overrightarrow{001,012} \}$.$\Box$
\end{example}

\begin{figure}[h]
		\includegraphics[scale=0.41]{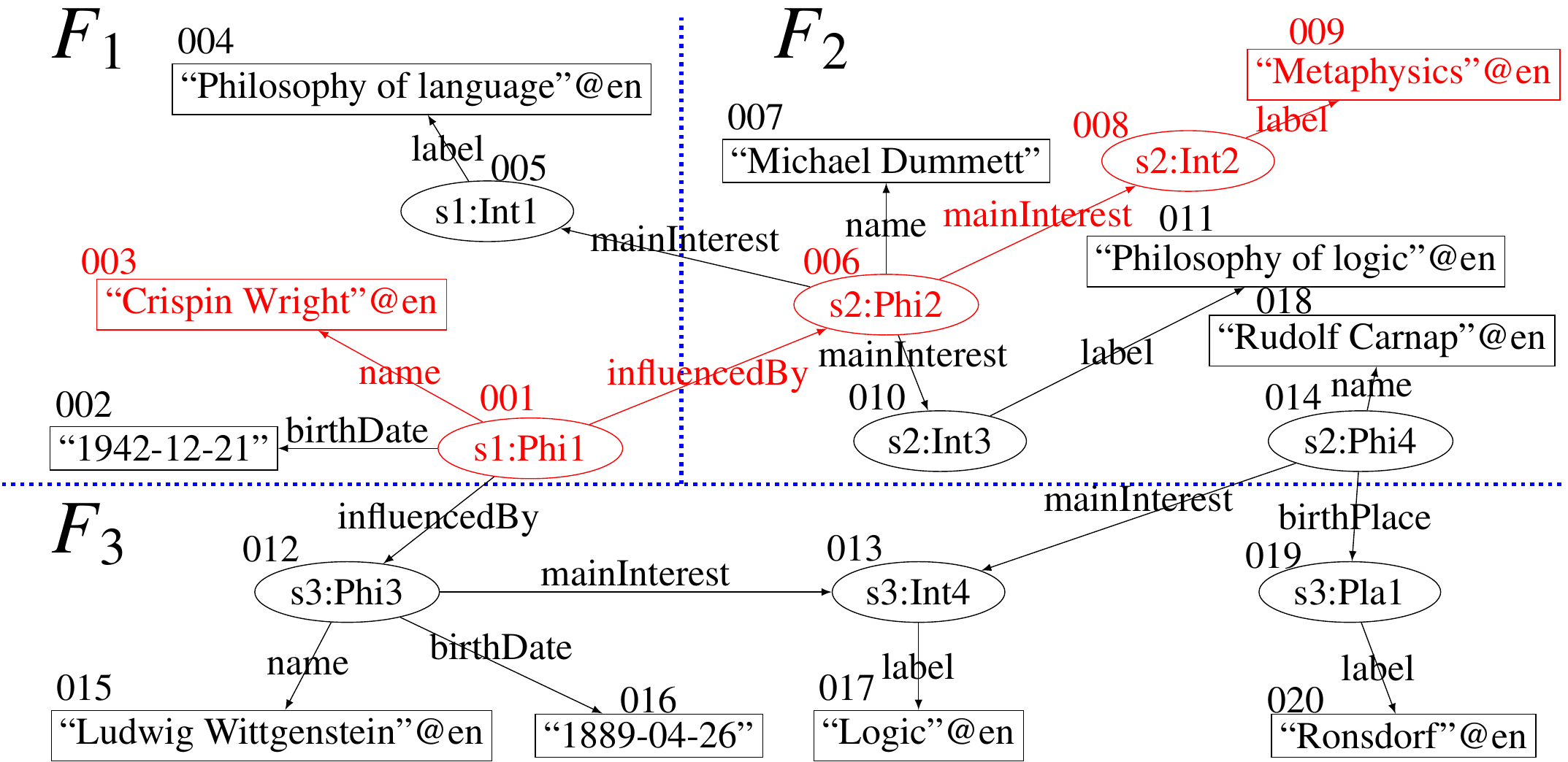}
\caption{Distributed RDF Graph}%
   \label{fig:datagraph}
\end{figure}

Similarly, a SPARQL query can also be represented as a query graph $Q$. In this study, we focus on BGP queries as they are foundational to SPARQL, and focus on techniques for handling them.

%Similarly, a SPARQL query can also be represented as a query graph $Q$. In this paper, we focus on basic graph pattern (BGP) queries as they are foundational to SPARQL, and focus on techniques for handling these.

\begin{definition}\label{def:query}\textbf{(SPARQL BGP Query)}
A \emph{SPARQL BGP query} is denoted as $Q=\{V^{Q},$ $E^{Q}, \Sigma^{Q}\}$, where $V^{Q} \subseteq V\cup V_{Var}$ is a set of vertices, $V$ denotes all vertices in the RDF graph $G$, $V_{Var}$ is a set of variables, and $E^{Q} \subseteq V^{Q} \times V^{Q}$ is a multiset of edges in $Q$. Each edge $e$ in $E^Q$ either has an edge label in $\Sigma$ (i.e., property), or the edge label is a variable.
\end{definition}

\begin{example}\label{example:query}
Fig. \ref{fig:querygraph} shows the query graph corresponding to the example query shown in Section \ref{sec:Introduction}. There are four edges in the query graph, and each edge maps to a triple pattern in the example query. Both vertices and edges in the query graph can be variable.
$\Box$
\end{example}

\begin{figure}[h]%\vspace{-0.15in}
		\includegraphics[scale=0.5]{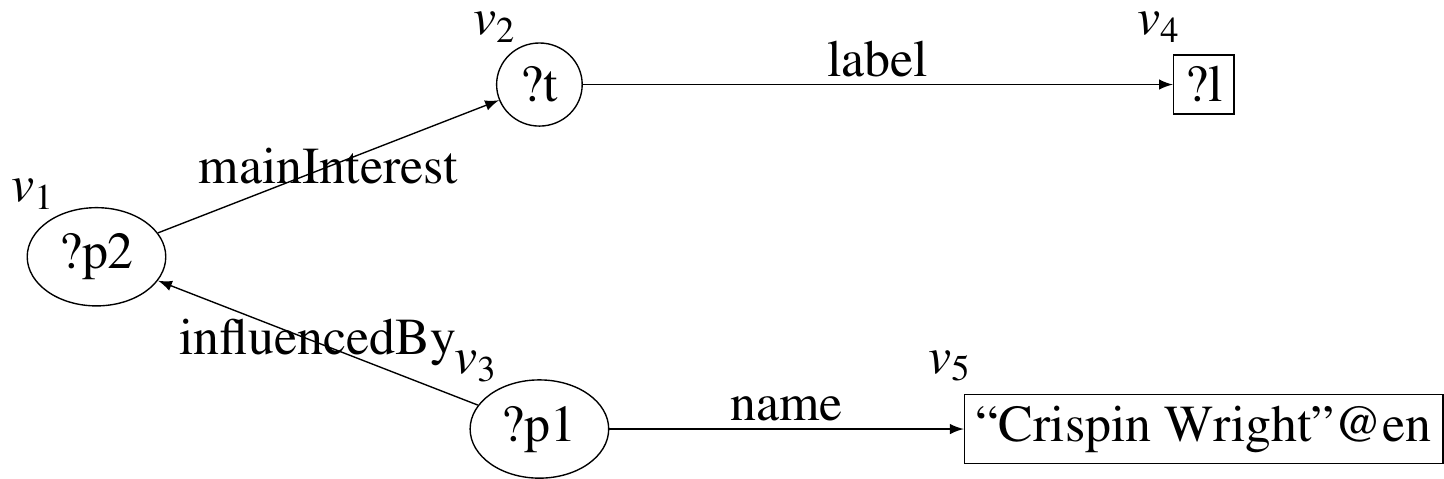}
%\vspace{-0.15in}
\caption{ SPARQL Query Graph}
%\vspace{-0.15in}
   \label{fig:querygraph}
\end{figure}

We assume that $Q$ is a connected graph; otherwise, all connected components of $Q$ are considered separately. Answering a SPARQL query is equivalent to finding all subgraphs of $G$ homomorphic to $Q$. The subgraphs of $G$ homomorphic to $Q$ are called \emph{matches} of $Q$ over $G$.

\begin{figure*}%\vspace{-0.1in}
\centering
		\includegraphics[scale=0.42]{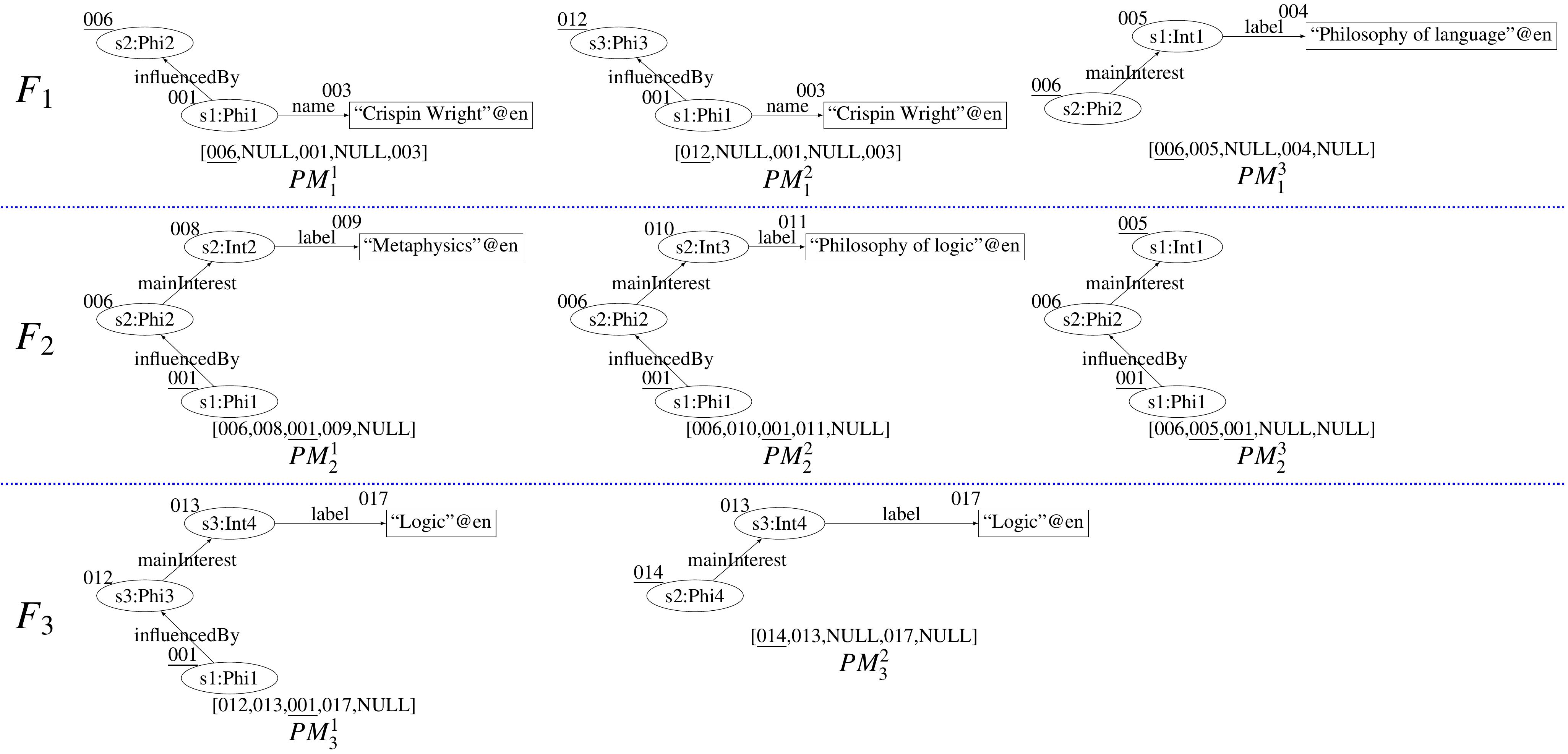}
%\vspace{-0.2in}
\caption{Example Local Partial Matches}
%\vspace{-0.275in}
   \label{fig:lpm}
\end{figure*}

\begin{definition}\label{def:sparqlmatch}\textbf{(SPARQL Match)}
Consider an RDF graph $G$ and a connected query graph $Q$ that has $n$ vertices $\{v_1,...,v_n\}$. A subgraph $M$ with $m$ vertices $\{u_1, ...,u_m\}$ (in $G$) is said to be a \emph{match} of $Q$ if and only if there exists a \emph{function} $f$ from $\{v_1,...,v_n\}$ to $\{u_1,...,u_m\}$ ($n \ge m$) where the following conditions hold:
if $v_i$ is not a variable, $f(v_i)$ and $v_i$ have the same uniform resource identifier (URI) or literal value ($1\leq i \leq n$);
if $v_i$ is a variable, there is no constraint over $f(v_i)$ except that $f(v_i)\in \{u_1,...,u_m\}$ ;
and if there exists an edge $\overrightarrow{v_iv_j}$ in $Q$, there also exists an edge $\overrightarrow{f{(v_i)}f{(v_j)}}$ in $G$. Let $L(\overrightarrow{v_iv_j})$ denote a multi-set of labels between $v_i$ and $v_j$ in $Q$, and $L(\overrightarrow{f{(v_i)}f{(v_j)}})$ denote a multi-set of labels between $f(v_i)$ and $f(v_j)$ in $G$. There must exist an \emph{injective function} from edge labels in $L(\overrightarrow{v_iv_j})$ to edge labels in $L(\overrightarrow{f{(v_i)}f{(v_j)}})$. Note that a variable edge label in $L(\overrightarrow{v_iv_j})$ can match any edge label in $L(\overrightarrow{f{(v_i)}f{(v_j)}})$.
\end{definition}

\begin{definition}{(\textbf{Problem Statement})}
Let $G$ be a distributed RDF graph that consists of a set of fragments $\mathcal{F} = \{F_{1}, \ldots, F_{k}\}$, and let $\mathcal{S} = \{S_{1}, \ldots, S_{k}\}$ be a set of sites such that $F_{i}$ is located at $S_{i}$. Given a SPARQL BGP query $Q$, our goal is to find all matches of $Q$ over $G$.
\end{definition}

Note that for simplicity of exposition, we are assuming that each site hosts one fragment.
Finding matches in a site can be evaluated locally using a centralized RDF triple store.
In this study, we only focus on how to find the matches crossing multiple sites efficiently. In our prototype experiments, we modify gStore \cite{Zou:2013fk} to perform partial evaluation.

\begin{example}\label{example:crossingmatch}
Given a SPARQL query graph $Q$ in Fig. \ref{fig:querygraph}, there exists a crossing match mapping to the subgraph induced by vertices 003,001,006,008, and 009 (shown in the red vertices and edges in Fig. \ref{fig:datagraph}).$\Box$
\end{example}

%\vspace{-0.05in}
\subsection{Partial Evaluation-Based SPARQL Query Evaluation}
\label{sec:LocalPartialMatch}
As we extend the distributed SPARQL query evaluation approach based on the ``partial evaluation and assembly'' framework in \cite{DBLP:journals/corr/PengZOCZ14}, we give its brief background here.

In our framework, each site $S_i$ receives the full query graph $Q$, and computes the partial answers (called \emph{local partial matches}) based on the known input $F_i$ (we assume that each site hosts one fragment, as indicated by its subscript). Intuitively, a local partial match $PM_i$ is an overlapping part between a crossing match $M$ and fragment $F_i$. Moreover, $M$ may or may not exist depending on the as-yet unavailable input ${G^\prime}$ . Based only on the known input $F_i$, we cannot judge whether $M$ exists or not.

\begin{definition}\label{def:localmaximal} \textbf{(Local Partial Match)} Given a SPARQL query graph $Q$ and a connected subgraph $PM$ with $n$ vertices $\{u_1,...,u_n\}$ ($n \leq |V^Q|$) in a fragment $F_k$, $PM$ is a \emph{local partial match} in $F_k$ if and only if there exists a function $f:V^Q$ $\to \{u_1,...,u_n\} \cup \{NULL\}$ that holds the following conditions:
\begin{enumerate}
  \item If $v_i$ is not a variable, $f(v_i)$ and $v_i$ have the same URI or literal value or $f(v_i)=NULL$.
\item If $v_i$ is a variable, $f(v_i) \in \{u_1,...,u_n\}$ or $f(v_i)=NULL$.
\item If there exists an edge $\overrightarrow{v_iv_j}$ in $Q$ ($ i\neq j$), then $PM$ should meet one of the following five conditions: there also exists an edge $\overrightarrow{f{(v_i)}f{(v_j)}}$ in $PM$ with property $p$ and $p$ is the same as the property of $\overrightarrow{v_iv_j}$, there also exists an edge $\overrightarrow{f{(v_i)}f{(v_j)}}$ in $PM$ with property $p$ and the property of $\overrightarrow{v_iv_j}$ is a variable, there does not exist an edge $\overrightarrow{f{(v_i)}f{(v_j)}}$ but $f{(v_i)}$ and $f{(v_j)}$ are both in $V_k^e$, $f{(v_i)}=NULL$, or $f{(v_j)}=NULL$.
\item  $PM$ contains at least one crossing edge, guaranteeing that an empty match does not qualify.
\item If $f(v_i) \in V_k$ (i.e., $f(v_i)$ is an internal vertex of $F_k$) and $\exists \overrightarrow {v_i v_j} \in Q$ (or $\overrightarrow {v_j v_i} \in Q$), there must exist $f(v_j) \neq NULL$ and  $\exists \overrightarrow {f(v_i)f(v_j)} \in PM$ (or $\exists \overrightarrow {f(v_j)f(v_i)} \in PM$). Furthermore, if $\overrightarrow {v_i v_j}$ (or $\overrightarrow {v_j v_i})$ has a property $p$, $\overrightarrow {f(v_i)f(v_j)}$ (or $\overrightarrow {f(v_j)f(v_i)}$) has the same property $p$.
\item If $f(v_i)$ and $f(v_j)$ are internal vertices of $F_k$, then there exist a \emph{weakly connected path} $\pi$ between $v_i$ and $v_j$ in $Q$ and each vertex in $\pi$ maps to an internal vertex of $F_k$.
\end{enumerate}

The vector $[ f{(v_1)}, ..., f{(v_n)}]$ is a serialization of a local partial match. $f^{-1}(PM)$ is the subgraph (of $Q$) induced by a set of vertices, where for any vertex $v\in f^{-1}(PM)$, $f(v)$ is not NULL.
\end{definition}

Generally, a local partial match is a subset of a complete SPARQL match. The first three conditions in Definition \ref{def:localmaximal} are analogous to a SPARQL match while vertices of query $Q$ are allowed to match a special value NULL.
The fourth condition requires that a local partial match must have at least one crossing edges, as it is used to form the possible crossing match.
The fifth condition is that if vertex $v$ (in query $Q$) is matched to an internal vertex, all neighbors of $v$ should also be matched in this local partial match.
The sixth condition is to ensure the correctness of our framework \cite{DBLP:journals/corr/PengZOCZ14}.

\begin{example}\label{example:pmset}
Given a query $Q$ in Fig. \ref{fig:querygraph} and a distributed RDF graph $G$ in Fig. \ref{fig:datagraph}, Fig. \ref{fig:lpm} shows all local partial matches and their serialization vectors in each fragment.
A local partial match in the fragment $F_i$ is denoted as $PM_i^j$, where the superscripts distinguish local partial matches in the same fragment. Furthermore, we underline all extended vertices in serialization vectors.

For example, $PM_1^1$ is the overlapping part between the crossing match discussed in Example \ref{example:crossingmatch} and fragment $F_1$. $PM_1^1$ contains a crossing edge $\overrightarrow{001,006}$. In $PM_1^1$, the query vertices $v_3$ and $v_5$ are matched to the internal vertices $001$ and $003$ of $F_1$, so $v_3$ and $v_5$ are weakly connected and all neighbors of $v_3$ and $v_5$ are also matched.
$\Box$
\end{example}

For a SPARQL query, local partial matches bear structural similarities (see Section \ref{sec:Property}); hence, they can be represented as vectors of Boolean formulas associated with crossing edges (see Section \ref{sec:DataStructure}). We can utilize these formulas to filter out some irrelevant local partial matches (see Section \ref{sec:LECBasedPruneRules}). Last, the remaining local partial matches are assembled to get the final answer (see Section \ref{sec:Assembly}). Note that,
in this study, we focus on how to represent the local partial matches in a compact way and prune some irrelevant local partial matches. We use the algorithm in \cite{DBLP:journals/corr/PengZOCZ14} directly to find local partial matches.

\section{Overview}\label{sec:Framework}
We extend the \emph{partial evaluation and assembly} \cite{DBLP:journals/csur/Jones96} framework to answer SPARQL queries over a distributed RDF graph $G$, as shown in Fig. \ref{fig:framework}. In our execution model, there are two stages: the partial evaluation stage and the assembly stage.

In the partial evaluation stage, each site $S_i$ first receives the full query graph $Q$ and finds all sets of internal candidates.The coordinator site assembles all sets of internal candidates from different sites, and gains the candidates' sets of all variables (Section \ref{sec:Advancing}). The coordinator site distributes the candidates' sets, and each site uses them to determine the local partial matches of $Q$ in $F_i$, at each site $S_i$. We explore the intrinsic structural similarities of local partial matches to divide these local partial matches into some equivalence classes, and propose a compact data structure named the \emph{LEC feature} (Definition \ref{def:LECSign}) to compress them. Only by joining LEC features can we determine the local partial matches that can contribute to the complete matches (Section \ref{sec:LPMIsomorphism}). In addition, we can also prove that the communication cost of all LEC features depends only on the size of the query and the partitioning of the graph (Section \ref{sec:Analysis}).

In the assembly stage, we divide all local partial matches into groups, and propose a join algorithm based on the LEC features (Section \ref{sec:Assembly}).

\begin{figure}[h]
%\vspace{-0.1in}
		\includegraphics[scale=0.18]{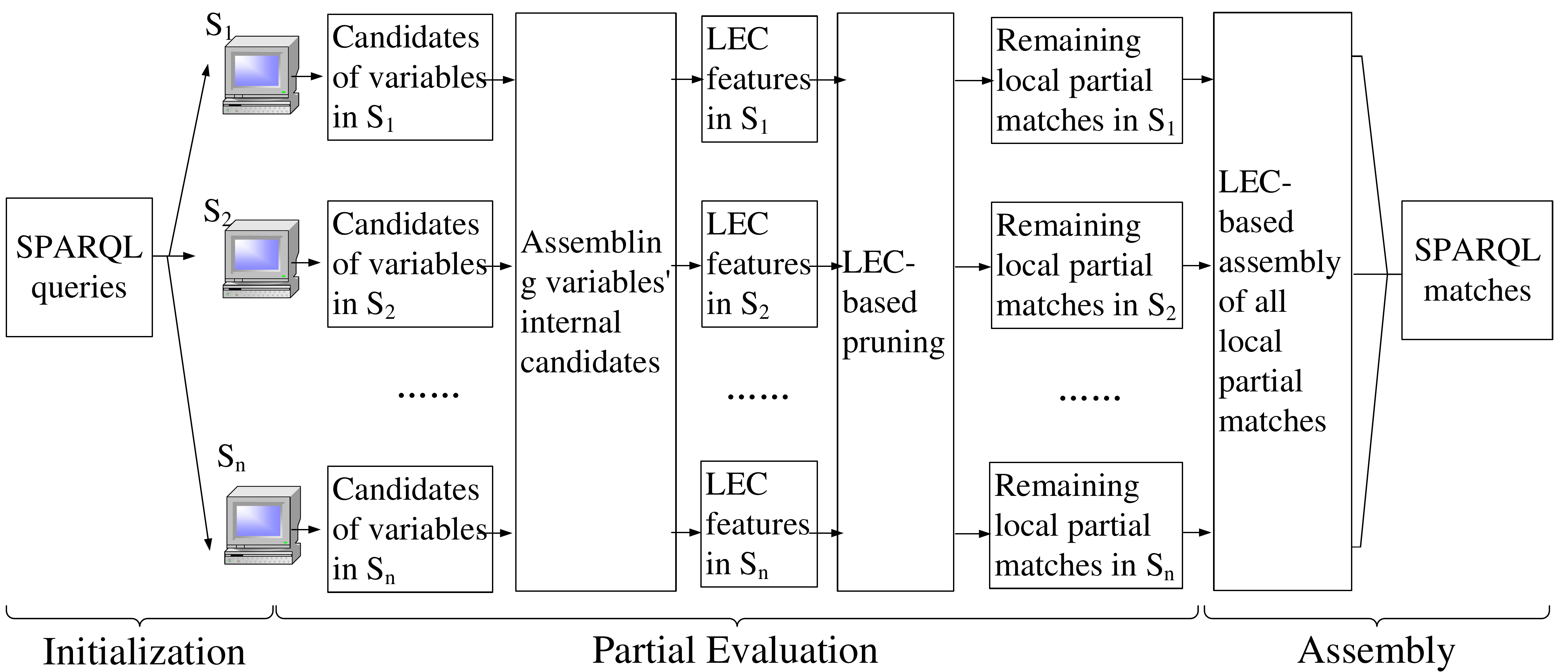}
%\vspace{-0.225in}
\caption{Overview of Our Method}
%\vspace{-0.25in}
   \label{fig:framework}
\end{figure}

\section{LEC Feature-based Optimization}\label{sec:LPMIsomorphism}
\subsection{Local Partial Match Equivalence Class}\label{sec:Property}
As discussed in \cite{DBLP:journals/corr/PengZOCZ14}, only local partial matches with common crossing edges from different fragments may join together via their common crossing edges. Hence, if two local partial matches generated from the same fragment contain the same crossing edges and these crossing edges map to the same query edges, then they can join with the same other local partial matches, and this means that they should have similar structures. For example, let us consider two local partial matches, $PM_2^1$ and $PM_2^2$ in Fig. \ref{fig:lpm}. They contain the common crossing edge $\overrightarrow {001,006}$, and $\overrightarrow {001,006}$  maps to the query edge $\overrightarrow {v_3v_1}$ in both $PM_2^1$ and $PM_2^2$. Thus, $PM_2^1$ and $PM_2^2$ are homomorphic to the same subgraph of the query graph. Any other local partial match (like $PM_1^1$) that can join with $PM_2^1$ can also join with $PM_2^2$.

We formalize the observation as the following theorem.

\begin{theorem}\label{theorem:LPMIsomorphism} Given two local partial matches $PM_i$ and $PM_j$ from fragment $F_k$ with functions $f_i$ and $f_j$, we can learn that $f_i^{-1}(PM_i)=f_j^{-1}(PM_j)$, where $f_i^{-1}(PM_i)$ and $f_i^{-1}(PM_i)$ are the subgraphs of $Q$ induced by the matched vertices, if they meet the following conditions:
\begin{enumerate}
\item $\forall \overrightarrow {u_iu_j}\in PM_i(or \ PM_j)$, if $\overrightarrow {u_iu_j}\in  E_k^c$, $\overrightarrow {u_iu_j}\in PM_j(or \ PM_i)$; and
  \item $\forall \overrightarrow {u_iu_j}\in PM_i(or \ PM_j)$, if $\overrightarrow {u_iu_j} \in E_k^c$, $f_i^{-1}(u_i) =f_j^{-1}(u_i)$ and $f_i^{-1}(u_j) =f_j^{-1}(u_j)$.
\end{enumerate}
\end{theorem}
\begin{proof}
First, we prove that $\forall v \in f_i^{-1}(PM_i)$, $v \in f_j^{-1}(PM_j)$. For any vertex $v \in f_i^{-1}(PM_i)$, there are two cases: 1) $PM_i$ contains an edge $e\in E_k^c$ and $f_i(v)$ is an endpoint of $e$; and 2) all edges adjacent to $f_i(v)$ in $PM_i$ are not crossing edges.

If $PM_i$ contains an edge $e\in E_k^c$ and $f_i(v)$ is an endpoint of $e$, $e \in E_k^c$, $e\in PM_j$. Hence, $f_i(v)\in PM_j$. Furthermore, because of condition 2, $v=f_i^{-1}(f_i(v))=f_j^{-1}(f_i(v))$. Thus, $v \in f_j^{-1}(PM_j)$.

Then, let us consider the case that all edges adjacent to $f_i(v)$ in $PM_i$ are not crossing edges. Because $f_i(v)$ does not belong to any crossing edges in $PM_i$, $f_i(v)$ is an internal vertex of $F_k$. According to condition 6 of Definition \ref{def:localmaximal}, there exists a weakly connected path between $v$ and any other vertices mapping to internal vertices in $PM_i$. Therefore, given a crossing edge $\overrightarrow {f_i(v_1)f_i(v_2)}\in PM_i$ where $f_i(v_1)$ is an internal vertex, there exists a weakly connected path $\pi = \{v_1, v_2,..., v\}$ in $f_i^{-1}(PM_i)$, and all vertices in $\pi$ map to internal vertices of $F_k$.

Let us consider the vertices in $\pi$ from $v_1$ to $v$ one by one. As $f_i(v_1)$ is an endpoint of a crossing edge, $v_1 \in f_j^{-1}(PM_j)$. In addition, because $PM_i$ and $PM_j$ are from the same fragment, $f_j(v_1)$ in $PM_j$ is still an internal vertex. According to condition 5 of Definition \ref{def:localmaximal}, all neighbors of $v_1$ have been matched in $PM_j$, so $v_2$ has been matched in $PM_j$. Furthermore, $f_j(v_2)$ must be an internal vertex. Otherwise, $\overrightarrow {f_j(v_1)f_j(v_2)}$ is a crossing edge, so $v_2 = f_j^{-1}(f_j(v_2)) =f_i^{-1}(f_j(v_2))$. In other words, $f_j(v_2)$ is an extended vertex of $F_k$ and also maps to $v_2$ in $f^{-1}_i(PM_i)$. This is in conflict with the fact that all vertices in $\pi$ map to internal vertices of $F_k$. By that analogy, we can prove that all other vertices in $\pi$ have been matched in $PM_j$. Hence, $v\in f_j^{-1}(PM_j)$ and $f_j(v)$ is an internal vertex.

Similarly, we can prove that $\forall v \in f_j^{-1}(PM_j)$, $v \in f_i^{-1}(PM_i)$. Therefore, the vertex set of $f_i^{-1}(PM_i)$ is equal to the vertex set of $f_j^{-1}(PM_j)$. Moreover, for each vertex $v$ in $f_i^{-1}(PM_i)$ and $f_j^{-1}(PM_j)$, both of $f_i(v)$ and $f_j(v)$ are internal vertices or extended vertices.

In contrast, for each edge $\overrightarrow {v_1v_2}\in f_i^{-1}(PM_i)$, owing to the condition 3 of Definition \ref{def:localmaximal}, at least one vertex of $f_i(v_1)$ and $f_i(v_2)$ is an internal vertex. Supposing that $f_i(v_1)$ is an internal vertex, $f_j(v_1)$ should also be an internal vertex, so $\overrightarrow {v_1v_2}\in f_j^{-1}(PM_j)$. In the same way, we can prove that $\forall \overrightarrow {v_1v_2} \in f_j^{-1}(PM_j)$, $\overrightarrow {v_1v_2} \in f_i^{-1}(PM_i)$. Hence, the edge set of $f_i^{-1}(PM_i)$ is equal to the edge set of $f_j^{-1}(PM_j)$.

In conclusion, $f_i^{-1}(PM_i)=f_j^{-1}(PM_j)$.
\end{proof}

Based on the above theorem, we can avoid exhaustive enumerations among irrelevant local partial matches with the same crossing edges that do not contribute to the final matches and result in significant data communication. Our strategy explores the intrinsic structural characteristics of the local partial matches only to generate combinations. If a generated combination cannot contribute to a valid match, we can filter out the local partial matches corresponding to the combination. To define the combination of multiple local partial matches, we first define the concept of a \emph{local partial match equivalence relation} as follows.

\begin{definition}\label{def:LPMEquivalenceRelation}\textbf{(Local Partial Match Equivalence Relation)}
Let $\Omega$ denote all local partial matches and $\sim$ be an equivalence relation over all local partial matches in $\Omega$ such that $PM_i$ $\sim$ $PM_j$ if $PM_i$(with function $f_i$) and $PM_j$(with function $f_j$) satisfy the following three conditions:
\begin{enumerate}
\item $PM_i$ and $PM_j$ are from the same fragment $F_k$;
\item $\forall \overrightarrow {u_iu_j}\in PM_i(or \ PM_j)$, if $\overrightarrow {u_iu_j}\in  E_k^c$, $\overrightarrow {u_iu_j}\in PM_j(or \ PM_i)$; and
  \item $\forall \overrightarrow {u_iu_j}\in PM_i(or \ PM_j)$, if $\overrightarrow {u_iu_j} \in E_k^c$, $f_i^{-1}(u_i) =f_j^{-1}(u_i)$ and $f_i^{-1}(u_j) =f_j^{-1}(u_j)$.
\end{enumerate}
\end{definition}

Based on the above equivalence relation, all local partial matches equivalent to a local partial match $PM_i$ can be combined together to form the \emph{L}ocal partial match \emph{E}quivalence \emph{C}lass (LEC) of $PM_i$ as follows.

\begin{definition}\label{def:LPMEquivalenceClass}\textbf{(Local Partial Match Equivalence Class)}
The \emph{local partial match equivalence class (LEC)} of a local partial match $PM_i$ is denoted $[PM_i]$, and is defined as the set
\[ [PM_i]=\{PM_j\in \Omega \ | \ PM_j \ \sim \ PM_i\}\]
\end{definition}

Then, we can prove that if two local partial matches can join together, then all other local partial matches in the corresponding LECs of the two local partial matches can also join together. Put another way, we only need to select one local partial match of a LEC as a representative to check whether all local partial matches in the LEC can join with other local partial matches. This prunes out many permutations of joining local partial matches of two LECs.

\begin{theorem}\label{theorem:LECEquivalence} Given two LECs $[PM_i]$ and $[PM_j]$, if a local partial match $PM_i$ can join with a local partial match $PM_j$, then any local partial matches in $[PM_i]$ can join with any local partial matches in $[PM_j]$.
\end{theorem}
\begin{proof}
As discussed in \cite{DBLP:journals/corr/PengZOCZ14}, if $PM_i$ and $PM_j$ can join together, then they are generated from different fragments, they share at least one common crossing edge that corresponds to the same query edge, and the same query vertex cannot be matched by different vertices in them.

Because $PM_i$ and $PM_j$ are from different fragments, according to Definition \ref{def:LPMEquivalenceRelation}, any local partial match in $[PM_i]$ is generated from different fragments from any local partial match in $[PM_j]$. Furthermore, all local partial matches in $[PM_i]$ (or $[PM_j]$) contain the same crossing edges that map to the same query edges, so any local partial match in $[PM_i]$ (or $[PM_j]$) shares at least one common crossing edge with any local partial match in $[PM_j]$ (or $[PM_i]$).

In addition, as our fragmentation is vertex-disjoint, the query vertices that the internal vertices in $PM_i$ map to should be different from the query vertices mapped to by the internal vertices in $PM_j$.
%In addition, since our fragmentation is vertex-disjoint, the query vertices that the internal vertices in $PM_i$ map to should be different from the query vertices that the internal vertices in $PM_j$ map to.
Hence, the internal vertices in any local partial match of $[PM_i]$ (or $[PM_j]$) cannot conflict with the internal vertices that any local partial match of $[PM_j]$ (or $[PM_i]$) map to. In addition, as the crossing edges in $PM_i$ does not conflict with the crossing edges in $PM_j$ and Definition \ref{def:LPMEquivalenceRelation} defines that the local partial matches in the same LEC share the same crossing edges and their mappings, the extended vertices in any local partial match of $[PM_i]$ (or $[PM_j]$) cannot conflict with the vertices that any local partial match of $[PM_j]$ (or $[PM_i]$) map to.

In summary, any two local partial matches in $[PM_j]$ and $[PM_i]$ meet all conditions that two joinable local partial matches should meet. Hence, the theorem is proven.
\end{proof}

\begin{example}\label{example:LECs}
Given all local partial matches in Fig \ref{fig:lpm}, there are seven LECs as follows.

$F_1:[PM_1^1]=\{PM_1^1\};\ [PM_1^2]=\{PM_1^2\},[PM_1^3]=\{PM_1^1\};$

$F_2:[PM_2^1]=[PM_2^2]=\{PM_2^1,PM_2^2\}, [PM_2^3]=\{PM_2^3\};$

$F_3:[PM_3^1]=\{PM_3^1\}, [PM_3^2]=\{PM_3^2\};$

As $PM_1^1$ can join with $PM_2^1$ and $PM_2^1$ and $PM_2^2$ are in the same LEC, $PM_1^1$ can also join with $PM_2^2$.$\Box$
\end{example}

%\vspace{-0.05in}
\subsection{LEC Feature}\label{sec:DataStructure}
Theorems \ref{theorem:LPMIsomorphism} and \ref{theorem:LECEquivalence} show that many local partial matches have the same structures, and can be combined together as a LEC to join with local partial matches of other LECs through their common crossing edges. The observations imply that we can only use the same structure of local partial matches in a LEC and the common crossing edges of the LEC to determine whether the local partial matches of the LEC can join with the local partial matches of other LECs.

Hence, given a LEC $[PM]$, we maintain it as a compact data structure called the \emph{LEC feature} that only contains the same structure of local partial matches in $[PM]$ and the common crossing edges of $[PM]$, as follows:
\begin{definition}\label{def:LECSign}\textbf{(LEC Feature)}
Given a local partial match $PM$ with function $f$ and its LEC $[PM]$, its \emph{LEC feature} $LF([PM])=\{F,g,LECSign\}$ consists of three components:
\begin{enumerate}
  \item The fragment identifier, $F$, that $PM$ is from;
  \item A function $g$, which maps crossing edge $\overrightarrow {u_iu_j}$ in $PM$ to its corresponding mapping $\overrightarrow {f^{-1}(u_i)f^{-1}(u_j)}$ in $E^Q$; and
  \item A bitstring of the length $|V^Q|$, $LECSign$, where the i-th bit is set to `1' if $f(v_i)$ maps to an internal vertex of $F$.
\end{enumerate}
\end{definition}

Fig. \ref{fig:LF} shows a LEC feature $LF([PM^1_1])$ for the LEC $[PM^1_1]$ shown in Example \ref{example:LECs}. In $LF([PM^1_1])$, $F_1$ is the fragment identifier of the fragment that $PM^1_1$ is generated from, and $\{\overrightarrow{001,006} \to \overrightarrow{v_3v_1}\}$ is the set of crossing edges in $PM^1_1$ and their corresponding query edges; as the internal vertices in $PM^1_1$ match the query vertices $v_3$ and $v_5$ that correspond to the third and fifth bits of $LECSign$, the $LECSign$ in $LF([PM^1_1])$ is $[00101]$.

\begin{example}\label{example:LFs}
Given the LECs in Example \ref{example:LECs}, their LEC features are as follows:

$LF([PM_1^1])=\{F_1,\{\overrightarrow{001,006} \to \overrightarrow{v_3v_1}\},[00101]\}$

$LF([PM_1^2])=\{F_1,\{\overrightarrow{001,012} \to \overrightarrow{v_3v_1}\},[00101]\}$

$LF([PM_1^3])=\{F_1,\{\overrightarrow{006,005} \to \overrightarrow{v_1v_2}\},[01010]\}$

$LF([PM_2^1])=LF([PM_2^2])=\{F_2,\{\overrightarrow{001,006} \to \overrightarrow{v_3v_1}\},$ $[11010]\} $

$LF([PM_2^3])=\{F_2,\{\overrightarrow{006,005} \to \overrightarrow{v_1v_2}, \overrightarrow{001,006} \to \overrightarrow{v_3v_1}\},$ $[10000]\} $

$LF([PM_3^1])=\{F_3,\{\overrightarrow{001,012} \to \overrightarrow{v_3v_1}\},[11010]\}$

$LF([PM_3^2])=\{F_3,\{\overrightarrow{014,013} \to \overrightarrow{v_1v_2}\},[01010]\}$
$\Box$
\normalsize
\end{example}

\begin{figure}[h]
%\vspace{-0.15in}
		\includegraphics[scale=0.425]{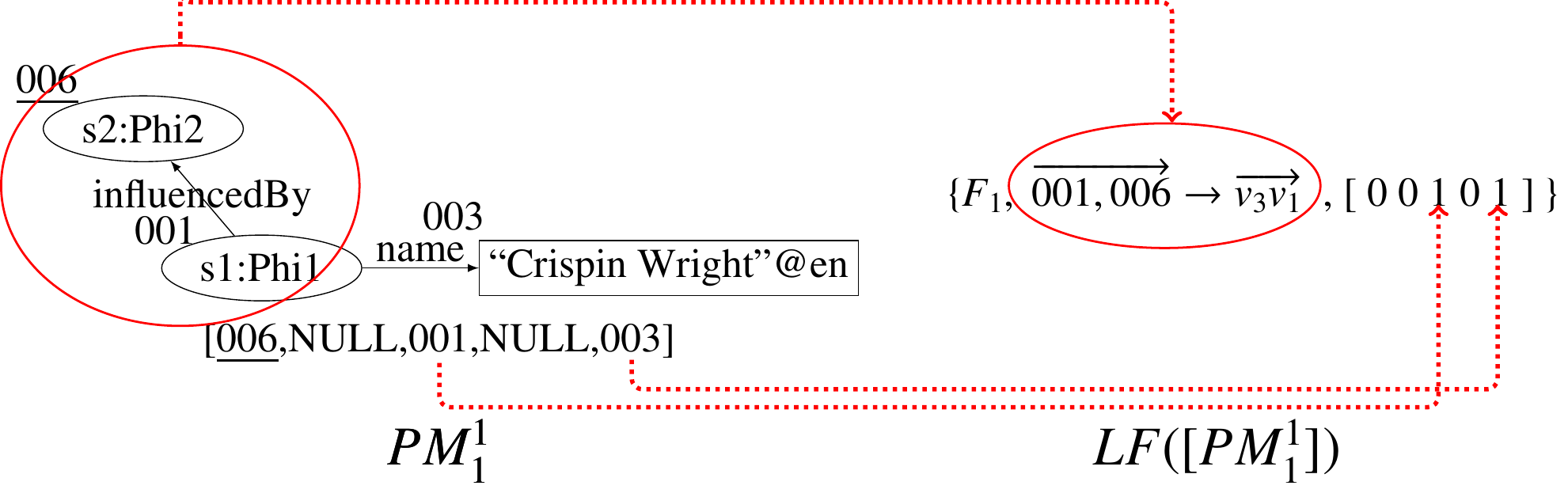}
%\vspace{-0.15in}
\caption{LEC Feature $LF([PM_1^1])$ ($PM_1^1$ is the only element in $[PM_1^1]$)}
%\vspace{-0.15in}
   \label{fig:LF}
\end{figure}

%\subsection{Computing LEC Features}\label{sec:ComputingLFeatures}
Given a SPARQL query $Q$ and a fragment $F_i$, we can find all LEC features (according to Definition \ref{def:localmaximal}) in $F_i$, and utilize them together to filter out some irrelevant local partial matches. In this study, we mainly focus on how to compress all local partial matches into LEC features. A high-level description of computing LEC features is outlined in Algorithm \ref{alg:findingLF}.

%\vspace{-0.1in}
\begin{algorithm}[h] \label{alg:findingLF}
\caption{Computing LEC Features}
\KwIn{The set of all local partial matches in fragment $F_i$, denoted as $\Omega(F_i)$.}
\KwOut{The set of all LEC features in $F_i$, denoted as $\Omega(F_i)$, denoted as $\Psi(F_i)$.}
\For{each local partial match $PM$ in $\Omega(F_i)$}{
Initialize a LEC feature $LF$;\\
$LF.F$ $\gets$ $F_i$;\\
\For{each mapping $(\overrightarrow {u_iu_j},\overrightarrow {v_iv_j})$ in $PM$}
{
    \If{$u_i$ (or $u_j$) is an extended vertex of fragment $F_i$}
     {
     $LF.LECSign[i]$ (or $LF.LECSign[j]$) $\gets$ `0';\\
     $LF.g$ $\gets$ $LF.g$ $\cup$ $(\overrightarrow {u_iu_j},\overrightarrow {v_iv_j})$;
     }
     \Else
     {
     $LF.LECSign[i]\gets$ `1' and $LF.LECSign[j]\gets$ `1';
     }
}
\If{$\Psi(F_i)$ does not contain $LF$}
 {
 $\Psi(F_i)$ $\gets$ $\Psi(F_i)$ $\cup$ $LF$;
 }
}
Return $\Omega(F_i)$;
\end{algorithm}
%\vspace{-0.15in}

The above process consists of determining what the LEC feature of a local partial match $PM$ is. We first initialize a LEC feature $LF$ with the fragment identifer $F_i$. Then, we scan all mappings in $PM$. For each mapping $(\overrightarrow {u_iu_j},\overrightarrow {v_iv_j})$, if $u_i$ (or $u_j$) is an extended vertex, we set $LF.LECSign[i]$ (or $LF.LECSign[j]$) as `0', and $LF.LECSign[i]$ (or $LF.LECSign[j]$) as `1' and add $(\overrightarrow {u_iu_j},\overrightarrow {v_iv_j})$ into $LF.g$. Last, we insert $LF$ into the set of all LEC features in $F_i$.
This above step iterates over each local partial match. Constructing all LEC features only requires a linear scan on the local partial matches; hence, it can be done on-the-fly because the local partial matches stream out from the evaluation.

%\vspace{-0.05in}
\subsection{LEC Feature-based Pruning Algorithm}\label{sec:LECBasedPruneRules}
In this section, based on the definition of the LEC feature and its properties, we propose an optimization technique that prunes some irrelevant local partial matches.

First, we define the conditions under which two local partial matches can join together as Definition \ref{def:JoinableLEC}, and prove the correctness of the join conditions as Theorem \ref{theorem:JoinableLEC}.

\begin{definition}\label{def:JoinableLEC}\textbf{(Joinable)}
Given two local partial matches $PM_i$ and $PM_j$, they are joinable if their LEC features $LF([PM_i])$ and $LF([PM_j])$ meet the following conditions:
\begin{enumerate}
  \item $LF([PM_i]).F \ne LF([PM_j]).F$;
  \item There exist at least one edge $\overrightarrow{u_iu_j}$, such that $LF([PM_i]).g(\overrightarrow{u_iu_j})$ $=LF([PM_j]).g(\overrightarrow{u_iu_j})$;
  \item There exist no two edges $\overrightarrow{u_iu_j}$ and $\overrightarrow{u_i^\prime u_j^\prime}$ in the domains of $LF([PM_i]).g$ and $LF([PM_j]).g$, respectively, such that $LF([PM_i]).g(\overrightarrow{u_iu_j})=LF([PM_j]).g(\overrightarrow{u_i^\prime u_j^\prime})$; and
  \item All bits in $LF([PM_i]).LECSign \land LF([PM_j]).LECSign$ are `0'.
\end{enumerate}
\end{definition}

\begin{theorem}\label{theorem:JoinableLEC} Given two LEC $[PM_i]$ and $[PM_j]$, if the LEC features of $[PM_i]$ and $[PM_j]$ are joinable, then any local partial match in $[PM_i]$ can join with any local partial match in $[PM_j]$.
\end{theorem}
\begin{proof}
Due to Condition 1 of Definition \ref{def:JoinableLEC}, any local partial match in $[PM_i]$ is generated from different fragments than any local partial match in $[PM_j]$ is generated from. Condition 2 of Definition \ref{def:JoinableLEC} means that any local partial matches in $[PM_i]$ shares at least one common crossing edge mapping to the same query edge with any local partial matches in $[PM_j]$. Condition 3 of Definition \ref{def:JoinableLEC} implies that the same query vertex cannot be matched by different vertices in crossing edges of local partial matches in $[PM_i]$ and $[PM_j]$. Condition 4 of Definition \ref{def:JoinableLEC} means that the same query vertex cannot be matched by different internal vertices edges of local partial matches in $[PM_i]$ and $[PM_j]$.

In summary, all conditions of Definition \ref{def:JoinableLEC} imply that all local partial matches in $[PM_i]$ and $[PM_j]$ meet all joining conditions discussed in \cite{DBLP:journals/corr/PengZOCZ14}. Hence, any local partial match in $[PM_i]$ can join with any local partial match in $[PM_j]$.
\end{proof}

Further, we prove in the following theorem that only by using all LEC features can we determine whether the local partial matches of a LEC can contribute to the complete matches.
\begin{theorem}\label{theorem:JoiningLECRes} Given $m$ ($m\le |V^Q|$) local partial matches $PM_1,PM_2,$ $..., PM_m$, the local partial matches can join together to form a match of $Q$ if their corresponding LEC features meet the following conditions:
\begin{enumerate}
  \item For any $PM_i$, there exists a local partial match $PM_j$ $(j\ne i)$ that $[PM_i]$ and $[PM_j]$ are joinable;
  \item $\forall 1\le i\ne j \le m$, all bits in $LF([PM_i]).LECSign\land LF([PM_j]).$ $LECSign$ are `0'; and
  \item All bits in $LF([PM_1]).LECSign\lor LF([PM_2]).LECSign$ $\lor ... \lor LF([PM_m]).LECSign$ are `1'.
\end{enumerate}
\end{theorem}
\begin{proof}
Here, we prove that if the three conditions in Theorem \ref{theorem:JoiningLECRes} are met, then $PM_1 \Join PM_2 \Join ... \Join PM_m$ is a match of $Q$.

Conditions 1 and 2 in Theorem \ref{theorem:JoiningLECRes} guarantees that the $m$ local partial matches can join together. Condition 3 in Theorem \ref{theorem:JoiningLECRes} means that each vertex $u$ in $PM_1 \Join PM_2 \Join ... \Join PM_m$ is an internal vertex of one local partial match $PM_i$ ($i\le m$).
As $u$ is an internal vertex in $PM_i$, all of $u$'s adjacent edges have been matched. Then, we can know all edges in $PM_1 \Join PM_2 \Join ... \Join PM_m$ have been matched. Hence, $PM_1 \Join PM_2 \Join ... \Join PM_m$ is a match of $Q$.
\end{proof}

Theorem \ref{theorem:JoiningLECRes} implies that we only need assemble all LEC features to determine which local partial matches can contribute to the complete match.
%If there exists a SPARQL match, there should be some LEC features that can be merged together and all bits in the union of these LEC features' $LECSign$ should be `1'.
Only when all bits in $LECSign$ of the joined result of some LEC features are `1' can the corresponding local partial matches join to form a SPARQL match.

Therefore, %before we assemble all local partial matches to form the complete matches,
we can assemble all LEC Features and merge them together to prune some irrelevant local partial matches. If a LEC feature cannot contribute to a union result of some LEC features' $LECSign$ where all bits are `1', then all local partial matches corresponding to the LEC feature can be pruned.

The straightforward approach of merging all LEC features is to check whether each pair of LEC features are joinable. However, the join space of the straightforward approach is very large; hence, we propose a partitioning-based optimized technique to reduce the join space. The intuition of our partitioning-based technique is that we divide all LEC features into multiple groups, such that two LEC features in the same group cannot be joinable. Then, we only consider joining LEC features from different groups.

\begin{theorem}\label{theorem:pmpartition} Given two LEC features $LF_i$ and $LF_j$, if $LF_i.LECSign=LF_j.LECSign$, $LF_i$ and $LF_j$ are not joinable.
\end{theorem}
\begin{proof}
Because $LF_i.LECSign$ is equal to $LF_j.LECSign$, $LF_i.$ $LECSign\land LF_j.LECSign = LF_i.LECSign = LF_j.$ $LECSign$. According to Condition 4 of Definition \ref{def:localmaximal}, there is at least one internal vertex in a local partial match, so there is at least one `1' in $LF_i.LECSign$ and $LF_j.LECSign$. Therefore, there is at least one `1' in $LF_i.LECSign\land LF_j.LECSign$, which is in conflict with Condition 4 of Definition \ref{def:JoinableLEC}.
\end{proof}

\begin{definition}\label{def:LFpartition} \textbf{(LEC Feature Group)}
Let $\Psi$ denote all LEC features. $\mathcal{P}=\{P_1,...,P_n\}$ is a set of \emph{LEC feature groups} for $\Psi$ if and only if each group $P_i$ ($i=1,...,n$) consists of a set of LEC features all having the same $LECSign$.
\end{definition}

\begin{example}\label{example:LFPartitioning}
Given all LEC features in Example \ref{example:LFs}, the LEC feature groups $\{P_1,P_2,P_3,P_4,P_5\}$ are as follows.

$P_1=\{LF([PM_1^1]),LF([PM_1^2])\}$,

$P_2=\{LF([PM_1^3])\},P_3=\{LF([PM_2^1]),LF([PM_3^1])\}$,

$P_4=\{LF([PM_2^3])\}, P_5=\{LF([PM_3^2])\}$.
$\Box$
\end{example}

Given a set $\mathcal{P}$ of LEC feature groups, we build a \emph{join graph} (denoted as $JG=\{V^{JG},E^{JG}\}$) as follows. In a join graph, one vertex indicates a LEC feature group. We introduce an edge between two vertices in the join graph if and only if some of their corresponding LEC feature groups can be joinable. Fig. \ref{fig:myjoingraph} shows the join graph of $\mathcal{P}$.

\begin{figure}[h]
%\vspace{-0.15in}
\centering
		\includegraphics[scale=0.5]{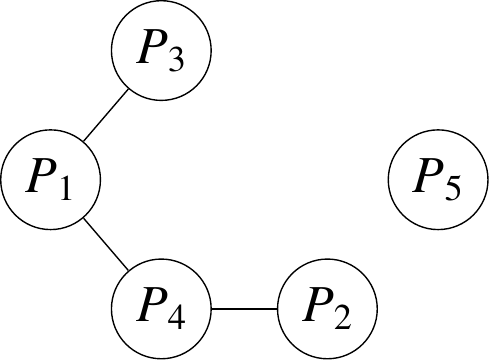}
%\vspace{-0.15in}
   \caption{Join Graph}
%   \vspace{-0.15in}
   \label{fig:myjoingraph}
\end{figure}

We propose an algorithm (Algorithm \ref{alg:LFPruning}) based on a depth-first search (DFS) traversal over the join graph, to filter out the irrelevant LEC features.
%We propose an algorithm (Algorithm ) based on the DFS traversal over the join graph to filter out the irrelevant LEC features.
For example, $P_5=LF([PM_3^2])$ in our example can be filtered out after we execute Algorithm \ref{alg:LFPruning}.

%\vspace{-0.15in}
\begin{algorithm}[h] \label{alg:LFPruning}
\caption{LEC Feature-based Pruning Algorithm}

\KwIn{A set $\mathcal{P}=\{P_{1},...,P_{n}\}$ of LEC feature groups and the join graph $JG$ }
\KwOut{The set $RS$ of LEC features that can contribute to complete matches}

$RS\gets \emptyset$;\\
\While{$V^{JG}\ne \emptyset$}{
    Find the vertex $v_{min}\in V^{JG}$ corresponding to LEC feature group $P_{min}$, where $P_{min}$ has the smallest size;\\
    Call Function \textbf{ComLECFJoin}($\{v_{min}\},P_{min},JG,RS$);\\
    Remove $v_{min}$ from $V^{JG}$;\\
    Remove all outliers remaining in $JG$;
}
\end{algorithm}

\begin{function}[h] \label{alg:functioncomLF}
\small
\caption{ComLECFJoin($V,P,JG,RS$)}
\For{each vertex $v$ in $JG$ adjacent to at least one vertex in $V$, where $v$ corresponds to LEC feature group $P^\prime$}
{
    Set $P^{\prime\prime}$ $\gets$  $\emptyset$; \\
    \For{each LEC feature $LF_i$ in $P$}
    {
      \For{each LEC feature $LF_j$ in $P^\prime$}
      {
      \If{$LF_i$ and $LF_j$ are joinable}
      {
         $LF_k\gets LF_i \Join LF_j$;\\
         \If{all bits in $LF_k.LECSign$ are `1'}
         {
         Insert all LEC features corresponding to vertices in $V$ into $RS$;
         }
         \Else
         {
         Put $LF_k$ into $P^{\prime\prime}$;
         }
       }
      }
    }

   Call Function \textbf{ComLECFJoin}($V\cup \{v\},P^{\prime\prime},JG$);\\
}
\end{function}

\subsection{Analysis}\label{sec:Analysis}
To analyze the complexity of the above optimization technique, we consider the communication and computation costs. The communication cost is the data shipment needed in the above optimization technique, whereas the computation cost is the response time needed for evaluating the query at different sites in parallel. In general, our method can guarantee the following:

\emph{\underline{Communication cost}}. As discussed previously, our optimization technique assembles the LEC features to prune the local partial matches. A general formula for determining the communication cost can be specified as
$Cost_{LF}\times |\Psi |$,
where $Cost_{LF}$ is the size of a LEC feature, and $|\Psi |$ is the number of LEC features.

For any LEC feature $\{F,g,LECSign\}$, its cost, $Cost_{LF}$, consists of three components. The first component is the cost of the fragment identifier $F$, which is a constant. The second component is the cost of the function $g$ mapping the crossing edges to the query edges. The number of crossing edges is at most $|E^Q|$, so the complexity of $g$ is $O(|E^Q|)$. The last component, $LECSign$, is defined as a bitstring of fixed-length $|V^Q|$, so the cost of $LECSign$ is also $Q(|V^Q|)$. In summary, the cost of any LEC feature is $O(|E^Q|+|V^Q|)$.

In contrast, the number of LEC features, $|\Psi |$, only depends on the number of crossing edges in fragment $F_i$ , i.e., $|E^c_i|$, because of the LEC features only introduced by these crossing edges. In the worst case, each query edge can map to any edge in $E^c_i$, and then the number of LEC features is $O(|E^c_i|^{|E^Q|})$. Hence, the number of LEC features is $O(\sum_{i=1}^{|\mathcal{F}|}|E^c_i|^{|E^Q|})$.

Overall, the total communication cost is $O(\sum_{i=1}^{|\mathcal{F}|}|E^c_i|^{|E^Q|}\times (|E^Q|+|V^Q|))$. Thus, given a distributed RDF graph $G$, our optimization technique has the property that the communication cost of evaluating a query depends mainly on the size of the query and the partitioning of the graph.

\emph{\underline{Computation cost}}.There are two parts of our optimization technique: partial evaluation for computing LEC features, and assembly for joining LEC features to obtain the final answer. We discuss the costs of the two stages as follows:

First, computing local partial matches to determine LEC features is performed on each fragment $F_i$ in parallel, and it takes $O(|V_i\cup V_i^e|^{|V^Q|})$ time to compute all local partial matches for each fragment. Hence, it takes at most $O(|V_m\cup V_m^e|^{|V^Q|})$ time to get all LEC features from all sites, where $V_m\cup V_m^e$ is the vertex set of the largest fragment in $\mathcal{F}$.

Second, we only need to scan all LEC features once to partition them, so it takes $O(|\Psi |)$ to partition all LEC features. In addition, given a partitioning $\mathcal{P}=\{P_1,...,P_n\}$, joining all LEC features costs $\prod\nolimits_{i = 1}^{i = n} {|P_{i} |}$, which is bounded by $O((\frac{|\Psi|}{|V^Q|})^{|V^Q|})$. As discussed previously, $|\Psi |$ independent of the entire graph $G$; hence, the response time is also independent of $G$.

\vspace{0.1in}
In summary, the data shipment of our method depends on the size of query graph and the number of crossing edges only, and the response time of our method depends only on the size of query graph, the largest fragment, and the number of edges across different fragments. Thus, our method is \emph{partition bounded} in both \emph{data shipment} and \emph{response time} \cite{DBLP:journals/pvldb/FanWWD14}.

\nop{
In real applications, we can expect that the number of crossing edges in a partitioning will be small compared to the size of the graph itself, i.e., $\sum_{i=1}^{|\mathcal{F}|}|E^c_i| \ll |V|$. Furthermore, after we study the real SPAQRL query workload, the DBpedia query workload\footnote{http://aksw.org/Projects/DBPSB.html}, the size of a real SPARQL query is often smaller than ten edges. Last, we find out that the query edge in real SPARQL queries often only map to a limited number of edges in $E^c_i$.
}

\section{LEC Feature-based Assembly}\label{sec:Assembly}
After we gain all local partial matches, we need to assemble and join all them to form all complete matches. In this section, we discuss the join-based assembly of local partial matches to compute the final results. The join method proposed in \cite{DBLP:journals/corr/PengZOCZ14} is a partitioning-based join algorithm, where the local partial matches are divided into multiple partitions based on their internal candidates, such that two local partial matches in the same partitions cannot be joinable. All local partial matches in the same partition map to the internal vertices for a given variable. In \cite{DBLP:journals/corr/PengZOCZ14}, the authors prove that the local partial matches in the same partition cannot be joined.

%After we gain all local partial matches, we need assemble and join all them to form all complete matches. In this section, we discuss the join-based assembly of local partial matches to compute the final results. The join method proposed in \cite{DBLP:journals/corr/PengZOCZ14} is a partitioning-based join algorithm, where the local partial matches are divided into multiple partitions based on their internal candidates such that two local partial matches in the same partitions cannot be joinable. All local partial matches in the same partition maps to the internal vertices for a given variable. In \cite{DBLP:journals/corr/PengZOCZ14}, the authors prove that the local partial matches in the same partition cannot be joined.
%\subsection{Local Partial Match Structure-based Optimizations} \label{sec:HashJoin}

The join space of the join algorithm in \cite{DBLP:journals/corr/PengZOCZ14} is still large. As discussed previously, we can determine whether two local partial matches in two different fragments can join according to their corresponding LEC features. Thus, we propose an optimized technique based on the LEC features of the local partial matches to join the local partial matches.

%The join space of the join algorithm in \cite{DBLP:journals/corr/PengZOCZ14} is still large. As discussed before, we can determine whether two local partial matches in two different fragments can join according to their corresponding LEC features. Thus, we proposes an optimized technique based on the LEC features of the local partial matches to join the local partial matches.

The intuition of our method is that we divide all local partial matches into multiple groups based on their LEC features as proved in Theorem \ref{theorem:pmpartition}, such that two local partial matches in the same group cannot be joinable. Then, we only consider joining local partial matches from different groups.

%The intuition of our method is that we divide all local partial matches into multiple groups based on their LEC features as proved in Theorem \ref{theorem:pmpartition} such that two local partial matches in the same group cannot be joinable. Then, we only consider joining local partial matches from different groups.

\nop{
\begin{theorem}\label{theorem:pmpartition} Given two local partial matches $PM_i$ and $PM_j$ and their LEC features $LF_i$ and $LF_j$, if $LF_i.LECSign$ is equal to $LF_j.LECSign$, $PM_i$ and $PM_j$ are not joinable.
\end{theorem}
\begin{proof}
Since $LF_i.LECSign$ is equal to $LF_j.LECSign$, $LF_i.$ $LECSign\land LF_j.LECSign = LF_i.LECSign = LF_i.$ $LECSign$. According to Condition 4 of Definition \ref{def:localmaximal}, there are at least one internal vertices in a local partial match, so there are at least one `1' in $LF_i.LECSign$ and $LF_i.LECSign$. Therefore, there are at least one `1' in $LF_i.LECSign\land LF_j.LECSign$, which is in conflict with Condition 4 of Definition \ref{def:JoinableLEC}.
\end{proof}
}

\begin{definition}\label{def:LFpartition} \textbf{(LEC Feature-based Local Partial Match Group)}
$\mathcal{G}=\{Gr_1,...,Gr_n\}$ is a set of local partial match groups for $\Omega $ if and only if each group $Gr_i$ ($i=1,...,n$) consists of a set of local partial matches and the corresponding LEC features of the local partial matches have the same $LECSign$.
\end{definition}

\begin{example}\label{example:LFPartitioning}
Given all local partial matches in Fig. \ref{fig:lpm}, after $PM_3^2$ is pruned during LEC feature-based optimization, the LECSign-based local partial match groups $\{Gr_1,Gr_2,Gr_3,Gr_4\}$ are as follows:

$Gr_1=\{PM_1^1,PM_1^2\},Gr_2=\{PM_1^3\}$,

$Gr_3=\{PM_2^1,PM_2^2,PM_3^1\},Gr_4=\{PM_2^3\}$
$\Box$
\end{example}

Given a set $\mathcal{G}$ of LECSign-based local partial match groups, we also build a \emph{local partial match group join graph} (denoted as $LG=\{V^{Gr},E^{Gr}\}$) as follows. In a join graph, one vertex indicates a LEC feature-based local partial match group. We introduce an edge between two vertices in the join graph if and only if some of their corresponding LEC features can be joinable.
Here, the join graph of $\mathcal{G}$ is shown in Fig. \ref{fig:joingraphLPMG}.

\begin{figure}[h]
\centering
		\includegraphics[scale=0.5]{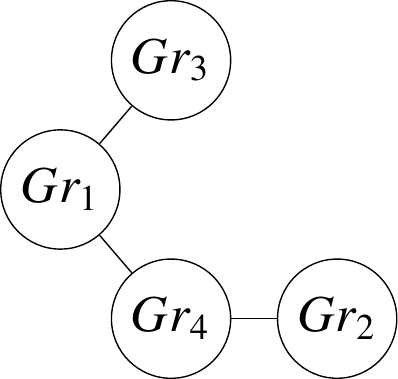}
   \caption{Local Partial Match Group Join Graph}
   \label{fig:joingraphLPMG}
\end{figure}

Then, we use Algorithm 3 based on the DFS traversal over the local partial match group join graph to get the complete matches.

\begin{algorithm}[h] \label{alg:advancedmerging}
\caption{LEC Feature-based Assembly Algorithm}

\KwIn{A set $\mathcal{G}=\{Gr_{1},...,Gr_{n}\}$ of LEC feature-based local partial match groups and its join graph $LG$ }
\KwOut{The set of complete matches, $MS$}
%Set $V$ $\gets$ $V^{JG}$, where $V^{JG}$ is the set of vertices in $JG$; \\
\While{$V^{Gr}\ne \emptyset$}{
    Find the vertex $v_{min}\in V^{Gr}$ corresponding to LEC feature-based local partial match group $Gr_{min}$, where $Gr_{min}$ has the smallest size;\\
    %\For{each neighbor of $v_{min}$ in $JG$}{
    Call Function \textbf{ComParJoin}($\{v_{min}\},Gr_{min},LG, MS$);\\
    %}
    Remove $v_{min}$ from $V^{Gr}$;\\
    Remove all outliers remaining in $LG$;
}
Return false;
\end{algorithm}

\begin{function}[h] \label{alg:functioncom}
\small
\caption{ComParJoin($V,Gr,LG, MS$)}
\For{each vertex $v$ in $LG$ adjacent to at least one vertex in $V$, where $v$ corresponds to $Gr^\prime$}
{
    $Gr^{\prime\prime}$ $\gets$  $\emptyset$; \\
    \For{each local partial match $PM_i$ in $Gr$}
    {
      \For{each local partial match $PM_j$ in $Gr^\prime$}
      {
      \If{$PM_i$ and $PM_j$ are joinable}
      {
         $PM_k\gets PM_i \Join PM_j$;\\
         \If{all vertices in $PM_k$ are matched}
         {
         Put $PM_k$ into $MS$;
         }
         \Else
         {
         Put $PM_k$ into $Gr^{\prime\prime}$;
         }
       }
      }
    }

   Call Function \textbf{ComParJoin}($V\cup \{v\},Gr^{\prime\prime},LG,MS$);\\
}
\end{function}

\nop{
\begin{example}\label{example:PartitioningAssemblyExample}
Let us consider the set of LECSign-based local partial match groups in Example \ref{example:LFPartitioning} and its corresponding join graph in Figure \ref{fig:joingraph}. Assume that we first select $P_4$ to call function ComParJoin to join with other LEC feature groups. In function ComParJoin, $P_4$ first join with $P_2$ and generate a new LEC feature group as $P^{*}=\{\{F^{*},\{\overrightarrow{001,006} \to \overrightarrow{v_3v_1}\},[11010]\}\}$, where $F^{*}$ is a fragment identifer different from any existing fragment identifer. Then, $P^{*}$ is selected as the parameter to call function ComParJoin and joins with LEC feature group $P_1$. After $P^{*}$ joins with $P_1$, the algorithm finds out the Boolean query answer and returns true.

If calling function ComParJoin for $P_4$ does not find out the final answer, $P_4$ is removed from the join graph. Then, $P_2$ becomes an outlier and is also removed from the join graph. Thus, the join graph only remains $P_1$ and $P_3$. We assume that $P_3$ is selected to call function ComParJoin. $P_3$ and $P_1$ will join together to find answers.
\end{example}
}

\section{Assembling Variables' Internal Candidates}\label{sec:Advancing}
In this section, we present another optimization technique: assembling variables' internal candidates. This technique is based on using the internal candidates of all variables in each site to filter out some false positives.

Existing RDF database systems used in sites storing individual fragments often adopt a filter-and-evaluate framework. They first compute out the candidates of all variables, and then search matches over these candidates. The process of finding candidates is often very quick. Hence, we can modify the codes of these systems and assemble the internal candidates in the coordinator site. When the set of internal candidates for variable $v$ (denoted as $C(Q,v)$) has been found, we do not find local partial matches directly, but send the set of internal candidates to the coordinator site.

The major benefit for assembling variables' internal candidates is to avoid some false positive local partial matches. When a site finds local partial matches, it does not consider how to join them with local partial matches in other sites. Hence, many unnecessary candidates may be generated, and they do not appear in any complete matches. To filter out these unnecessary candidates, the coordinator site can assemble and union the candidates' sets of a variable from all sites. If a candidate of variable $v$ can appear in a complete match, it must belong to $v$'s internal candidate sets from all sites. Then, when we compute the local partial matches, we avoid forming the local partial matches over those extended candidates that do not appear in the assembled internal candidates.

In practice, there may be too many internal candidates for each variable, resulting in a high communication cost. To reduce the communication cost, we compress the information of all internal candidates for each variable into a fixed-length bit vector. For variable $v$, we associate it with a fixed length bit vector $B_v$. We define a hash function to map each of $v$'s internal candidates in a site to a bit in $B_v$. Then, all $v$'s internal candidates can be compressed in $B_v$.
Thus, the coordinator site only needs to assemble all bit vectors of variables from different sites and to perform bitwise OR operations over bit vectors of a variable from different sites. We can send the result bit vectors of all variables to different sites and filter out some false positive candidates. Because the length of a bit vector is fixed, the communication cost is not too expensive.

Smaller search space can speed up evaluating the SPARQL query, meanwhile modern distributed environments have much faster communication networks than in the past. Hence, it is beneficial for us to afford the cost of communicating the candidate bit vectors between the coordinator site and the sites.

Algorithm 4 describes the optimization of assembling variables' internal candidates. The coordinator site receives and unions the bit vectors of candidates of all variables. Then, the coordinator site sends the result bit vectors of all variables to sites. For each site, it firstly finds out the candidates of variables locally, and compresses them into bit vectors. It then sends all bit vectors to the coordinator site, and waits for the bit vectors of all variables from the coordinator site.

With the received bit vectors of all variables, the site can filter out many false positive extended candidates during the computing of the local partial matches.

\begin{algorithm} \label{alg:putting}
\caption{Assembling Variables' Internal Candidates}
\KwIn{A distributed RDF graph $G$ over sites $\{S_1, ..., S_m\}$, coordinator site $S_c$, and the SPARQL query $Q$.}
\KwOut{ The internal bit vector $B_v$ of any variable $v$.}
\noindent\textbf{The Coordinator Site $S_c$:} \\
\For{each variable $v$ in $Q$}
{
    $B_v\gets 0$;\\
    \For{each site $S_{i}$}
    {
        Receive $B_v^\prime$ from $S_{i}$;\\
        $B_v\gets B_v\lor B_v^\prime$;\\
    }

    \For{each site $S_{i}$}{
        Send $B_v$ to $S_{i}$;\\
    }
}

\noindent\textbf{The Site $S_{i}$:} \\
\For{each variable $v$ in $Q$}
{
    Find $C(Q,v)$ and $B_v^\prime\gets 0$;\\
    \For{each internal candidate $c$ in $C(Q,v)$}{
        Use a hash function $h$ to map $c$ to an integer $h(c)$;\\
        Set the $h(c)$-th bit of $B_v^\prime$ to 1;\\
    }
    Send $B_v$ to $S_c$;\\
    Receive $B_v$ from $S_c$;\\
}
\end{algorithm}

\section{Impact of Partitioning Strategies}\label{sec:Partitioning}
In this section, we analyze the impact of different partitioning strategies for our method.

According to the above analysis, the costs of our method are mainly dependent on the number of LEC features. The straightforward heuristic is to reduce the number of crossing edges. However, if we examine the complexity of the cost more deeply, we discover that the small size of an edge cut does not always result in a small number of LEC features. For example, let us consider two example partitionings in Fig. \ref{fig:PartitioningExample}. Although the partitioning in Fig. \ref{fig:EISPartitioning} results in more crossing edges, its crossing edges are scattered to different boundary vertices. In contrast, all crossing edges in Fig. \ref{fig:METISPartitioning} are adjacent to one boundary vertex. When a star query $Q$ of two edges
as Fig. \ref{fig:PartitioningQuery}
is input, it maps to $\left( {\begin{array}{c}
   4  \\
   2 \\
\end{array}} \right)+\left( {\begin{array}{c}
   4  \\
   1 \\
\end{array}} \right)=10
$ LEC features for the partitioning in Fig. \ref{fig:METISPartitioning}, and $\left( {\begin{array}{c}
   3  \\
   2 \\
\end{array}} \right)+\left( {\begin{array}{c}
   3  \\
   1 \\
\end{array}} \right)+\left( {\begin{array}{c}
   2  \\
   2 \\
\end{array}} \right)+\left( {\begin{array}{c}
   2  \\
   1 \\
\end{array}} \right)=9
$ LEC features for the partitioning in Fig. \ref{fig:EISPartitioning}.

\begin{figure}[h]
   \centering
\subfigure[{ }]{%
		\includegraphics[scale=0.25]{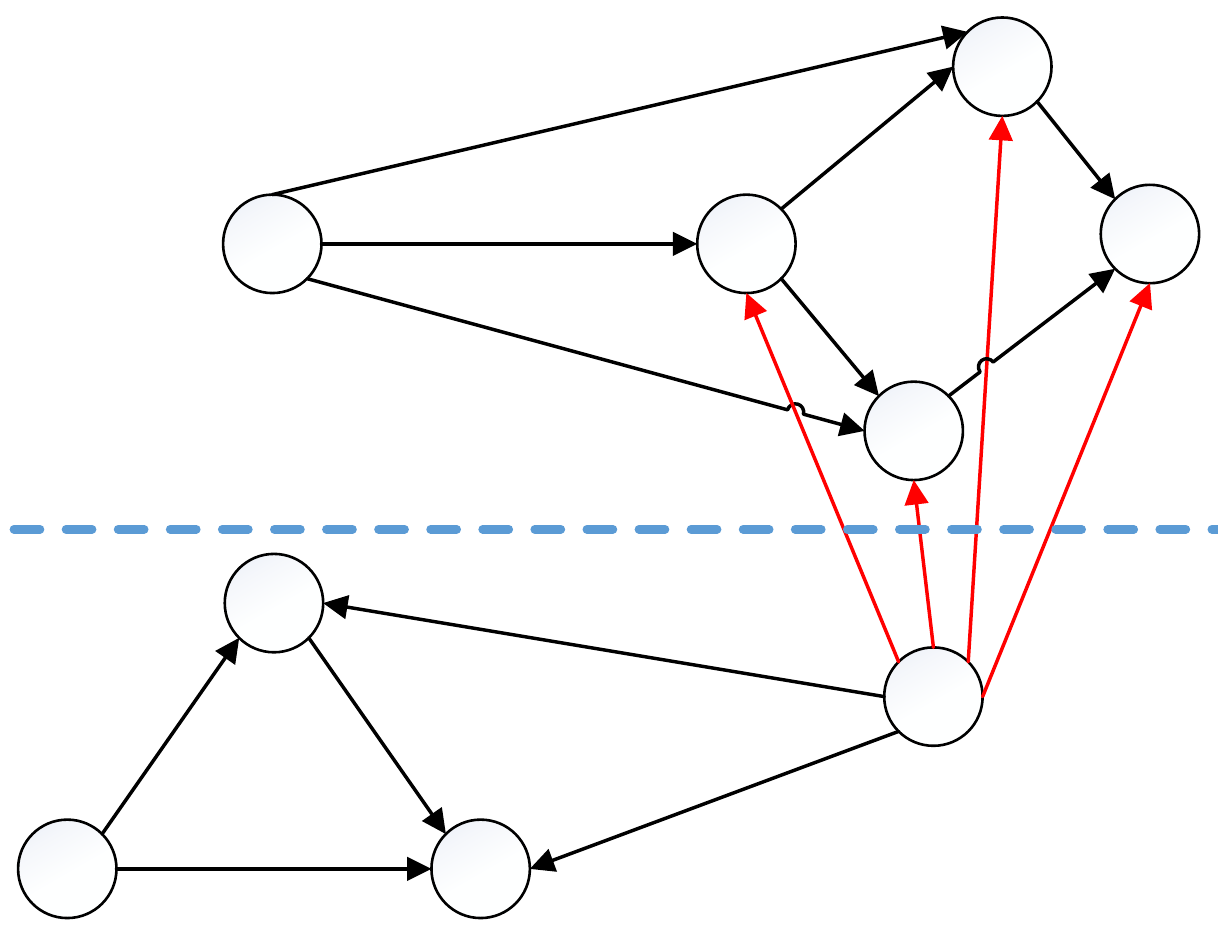}
       \label{fig:METISPartitioning}%
       }%
   \subfigure[{}]{%
		\includegraphics[scale=0.25]{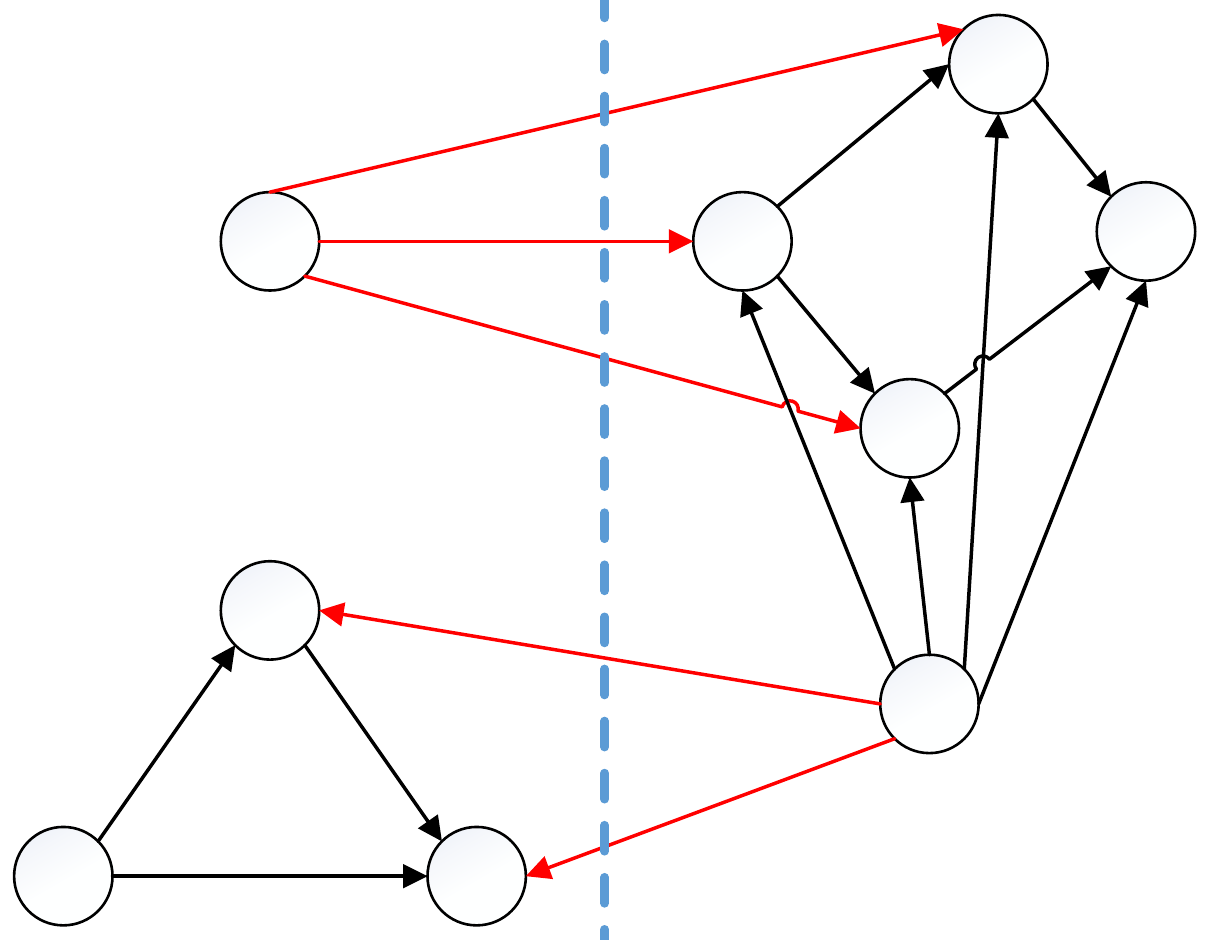}
       \label{fig:EISPartitioning}%
       }
\subfigure[{}]{%
		\includegraphics[scale=0.25]{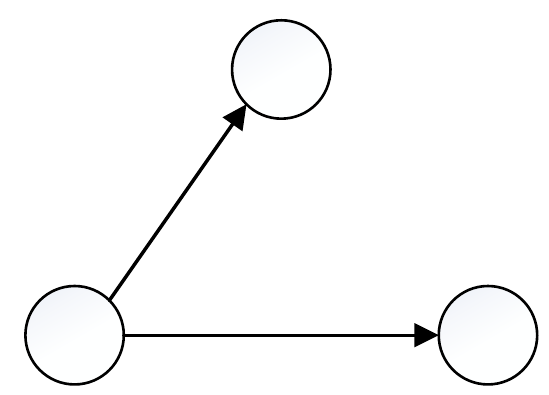}
       \label{fig:PartitioningQuery}%
       }
 \caption{ Comparison of Different Partitionings}%
 \label{fig:PartitioningExample}
\end{figure}

Based on the above observation, in a good partitioning for our method, the crossing edges need to be scattered to as many vertices as possible. Given a partitioning $\mathcal{F} = \{F_1,F_2,...,F_k\}$ and the
set of its crossing edges $E^c$, we define the distribution of crossing edges over a vertex $v$, $p_{\mathcal{F}}(v)$, as follows.
\[p_{\mathcal{F}}(v)=|N(v)\cap E^c|/(2\times |E^c|)\]

In the above, $N(v)$ is the set of $v$'s neighbors. Note that, an edge can be adjacent to two vertices, so the divisor in $p_{\mathcal{F}}(v)$ should be $2\times |E^c|$, which can ensure that the sum of the distributions over all vertices is $1$.

Then, the expectation of the number of crossing edges adjacent to a vertex $v$ is as follows.
\[E_{\mathcal{F}}(v)=|N(v)\cap E^c|\times p_{\mathcal{F}}(v)\]

Then, the total expectation of the number of crossing edges distributed to all vertices is as follows.
\[E_{\mathcal{F}}(V)=\sum_{v\in V}E_{\mathcal{F}}(v)\]

To scatter the crossing edges to as many vertices as possible, the above expectation should be as small as possible.

In addition, when we partition the graph, we should also balance the sizes of all fragments. Thus, we should avoid generating a fragment with too many edges. Here, we use the edge number of the largest fragment to measure the balance of fragments. In summary, we combine the above two factors to define the cost of a partitioning as follows.
\[Cost_{Partitioning}(\mathcal{F})=E_{\mathcal{F}}(V)\times \max_{1\le i\le k}|E_i\cup E_i^c|\]

%In real applications, we select the partitioning with the smallest above cost.
Here, a more sophisticated partitioning strategy is beyond the scope of this study. We only select the partitioning with the smallest cost from the existing partitioning strategies.
For example, the cost of the partitioning in Fig. \ref{fig:METISPartitioning} is $27.5$, and the cost of the partitioning in Fig. \ref{fig:EISPartitioning} is $23.4$. Hence, the partitioning in Fig. \ref{fig:EISPartitioning} is a better partitioning to be selected.

\section{Experiments}\label{sec:Experiment}
In this section, we use some real and synthetic RDF datasets to conduct our experiments.% All benchmark queries are listed in Appendix.

\subsection{Setting}
\textbf{LUBM}. LUBM \cite{DBLP:LUBM} is a benchmark that adopts an ontology for the university domain,
and can generate RDF data scalable to an arbitrary size. We generate three datasets of triples from 100 million to 1 billion, whose sizes vary from 15 GB to 150 GB. The dataset of 100 million triples is denoted as LUBM 100M, the one of 500 million triples is LUBM 500M and the one of 1 billion triples is LUBM 1B.
We use the 7 benchmark queries in \cite{DBLP:journals/pvldb/AbdelazizHKK17} (denoted as $LQ_1-LQ_7$) to test our methods.

%We also rewrite the benchmark queries by replacing the select clauses with the ask clauses.

\textbf{YAGO2}. YAGO2 \cite{DBLP:YAGO2} is a real RDF dataset that is extracted from Wikipedia.
YAGO2 also integrates its facts with the WordNet thesaurus.
It contains approximately $284$ million triples of 44 GB. We use the benchmark queries in \cite{DBLP:journals/pvldb/AbdelazizHKK17} (denoted as $YQ_1-YQ_4$) to evaluate our methods.

\textbf{BTC}. BTC\footnote{http://km.aifb.kit.edu/projects/btc-2012/}
is a real dataset used for the
Billion Triples Track of the Semantic Web Challenge, and contains approximately
1 billion triples of 176 GB. We use the 7 queries (denoted as $BQ_1-BQ_7$) in \cite{DBLP:journals/corr/PengZOCZ14} to test our methods.

\begin{table*}
  \caption{Evaluation of Each Stage on LUBM 100M}
\scriptsize
\begin{threeparttable}
  \begin{tabular}{|p{0.24cm}|p{0.06cm}|r|r|r|r|r|r|r|r|r|r|r|}
  \hline
     &  & \multicolumn{6}{c|}{Partial Evaluation}  & \multicolumn{1}{c|}{{Assembly}}  &  & &  &\\
  \cline{3-9}
   &  & \multicolumn{2}{l|}{\tabincell{p{2.5cm}}{Assembling Variables' Internal Candidates}} & \tabincell{p{1.4cm}}{Time of Local Partial Match}& \multicolumn{2}{l|}{\tabincell{p{2.5cm}}{LEC Feature-based Optimization}} & & \multicolumn{1}{l|}{\tabincell{p{1.4cm}}{ Time of LEC Feature-based}}  & \tabincell{p{1.1cm}}{Total Time (in ms)} & \tabincell{p{0.9cm}}{Local Partial}  &  \tabincell{p{0.9cm}}{Matches' Number} & \tabincell{p{1.0cm}}{Crossing Matches'} \\
  \cline{3-4}\cline{6-7}
    & &  \tabincell{p{0.75cm}}{Time (in ms)} &  \tabincell{p{1.5cm}}{Data Shipment (in KB)} &\tabincell{p{1.4cm}}{Computation (in ms)} & \tabincell{p{0.75cm}}{Time (in ms)} & \tabincell{p{1.5cm}}{Data Shipment (in KB)}	& Time(in ms) & \tabincell{p{1.4cm}}{Assembly ~~~~~~(in ms)} & &\tabincell{p{0.9cm}}{Matches' Number} & & \tabincell{p{1.0cm}}{Number ~~~~~~}\\
  \hline
   \tabincell{l}{$LQ_1$} &  & 4,029	& 2,032	& 21,550	& 2,054	& 38,882	& 27,633 & 12,539	& 40,172	& 276,327	& 21	& 21\\
  \cline{1-13}
    \tabincell{l}{$LQ_2$} &  &0 &  0 &8,488  & 0 &  0 &  0	&  0 &  8,488 & 0 & 864,197   & 0 \\
  \hline
  \tabincell{l}{$LQ_3$} &  & 568	& 16	& 2,795	& 0	& 0	& 3,363 & 0	& 3,363	& 0	& 0	& 0\\
  \cline{1-13}
   \tabincell{l}{$LQ_4$} & \tabincell{l}{$\surd$} & 0 & 0  & 221 &0  & 0  &  0	&  0 &  221 & 0 &  10  & 0\\
  \cline{1-13}
   \tabincell{l}{$LQ_5$} & \tabincell{l}{$\surd$} & 0 &  0 & 187 & 0 &  0 &  0	&  0 &  187 & 0 &  10  &0\\
   \cline{1-13}
   \tabincell{l}{$LQ_6$} & \tabincell{l}{$\surd$} & 1,556	& 136	& 1,516	& 61	& 1	& 3,133 & 9	& 3,142	& 228	& 125	& 114\\
  \cline{1-13}
    \tabincell{l}{$LQ_7$} &  & 7,827 & 	2,268	& 25,779	& 2,323	& 5,057	& 35,929 &  12,582	& 48,511	& 973,255	& 35,434	& 35,077\\
  \hline
  \end{tabular}
  \begin{tablenotes}
  \scriptsize
  \item[] $\surd$ means that the query involves some selective triple patterns.
  \end{tablenotes}
  \end{threeparttable}
  \label{table:queriesperformanceLUBM}
\end{table*}

\begin{table*}
  \caption{Evaluation of Each Stage on YAGO2}
\scriptsize
\begin{threeparttable}
  \begin{tabular}{|p{0.24cm}|p{0.06cm}|r|r|r|r|r|r|r|r|r|r|r|}
  \hline
     &  & \multicolumn{6}{c|}{Partial Evaluation}  & \multicolumn{1}{c|}{{Assembly}}  &  & &  &\\
  \cline{3-9}
   &  & \multicolumn{2}{l|}{\tabincell{p{2.5cm}}{Assembling Variables' Internal Candidates}} & \tabincell{p{1.4cm}}{Time of Local Partial Match}& \multicolumn{2}{l|}{\tabincell{p{2.5cm}}{LEC Feature-based Optimization}} & & \multicolumn{1}{l|}{\tabincell{p{1.4cm}}{ Time of LEC Feature-based}}  & \tabincell{p{1.1cm}}{Total Time (in ms)} & \tabincell{p{0.9cm}}{Local Partial}  &  \tabincell{p{0.9cm}}{Matches' Number} & \tabincell{p{1.0cm}}{Crossing Matches'} \\
  \cline{3-4}\cline{6-7}
    & &  \tabincell{p{0.75cm}}{Time (in ms)} &  \tabincell{p{1.5cm}}{Data Shipment (in KB)} &\tabincell{p{1.4cm}}{Computation (in ms)} & \tabincell{p{0.75cm}}{Time (in ms)} & \tabincell{p{1.5cm}}{Data Shipment (in KB)}	& Time(in ms) & \tabincell{p{1.4cm}}{Assembly ~~~~~~(in ms)} & &\tabincell{p{0.9cm}}{Matches' Number} & & \tabincell{p{1.0cm}}{Number ~~~~~~}\\
  \hline
   \tabincell{l}{$YQ_1$} &  & 188	&  13	&  1,007	&  879	&  6	&  2,094 & 79	& 2,153	& 811	& 17	& 17\\
  \cline{1-13}
    \tabincell{l}{$YQ_2$} &  & 315	& 15	& 999	& 26	& 1	& 1,340& 0	&1,340	&0	&0	&0\\
  \hline
  \tabincell{l}{$YQ_3$} &  & 1,341	& 137	& 3,292	& 1,599	& 1,317	& 6,232 & 21,404	& 27,636	& 816,382	& 605,993	& 588,390\\
  \cline{1-13}
   \tabincell{l}{$YQ_4$} &  &  388	& 27	& 2,036	& 1,602	& 293	& 4,026 & 686	& 4,712	& 16,661	& 226	& 224\\
  \hline
  \end{tabular}
  \end{threeparttable}
  \label{table:queriesperformanceYAGO2}
\end{table*}

\begin{table*}
\caption{Evaluation of Each Stage on BTC}
\scriptsize
\begin{threeparttable}
  \begin{tabular}{|p{0.24cm}|p{0.06cm}|r|r|r|r|r|r|r|r|r|r|r|}
  \hline
     &  & \multicolumn{6}{c|}{Partial Evaluation}  & \multicolumn{1}{c|}{{Assembly}}  &  & &  &\\
  \cline{3-9}
   &  & \multicolumn{2}{l|}{\tabincell{p{2.5cm}}{Assembling Variables' Internal Candidates}} & \tabincell{p{1.4cm}}{Time of Local Partial Match}& \multicolumn{2}{l|}{\tabincell{p{2.5cm}}{LEC Feature-based Optimization}} & & \multicolumn{1}{l|}{\tabincell{p{1.4cm}}{ Time of LEC Feature-based}}  & \tabincell{p{1.1cm}}{Total Time (in ms)} & \tabincell{p{0.9cm}}{Local Partial}  &  \tabincell{p{0.9cm}}{Matches' Number} & \tabincell{p{1.0cm}}{Crossing Matches'} \\
  \cline{3-4}\cline{6-7}
    & &  \tabincell{p{0.75cm}}{Time (in ms)} &  \tabincell{p{1.5cm}}{Data Shipment (in KB)} &\tabincell{p{1.4cm}}{Computation (in ms)} & \tabincell{p{0.75cm}}{Time (in ms)} & \tabincell{p{1.5cm}}{Data Shipment (in KB)}	& Time(in ms) & \tabincell{p{1.4cm}}{Assembly ~~~~~~(in ms)} & &\tabincell{p{0.9cm}}{Matches' Number} & & \tabincell{p{1.0cm}}{Number ~~~~~~}\\
  \hline
   \tabincell{l}{$BQ_1$} & \tabincell{l}{$\surd$} & 0 &  0 &259  & 0 &  0 &  0	&  0 &  259 & 0 & 1   & 0\\
  \cline{1-13}
    \tabincell{l}{$BQ_2$} &  \tabincell{l}{$\surd$}&0 &  0 &269  & 0 &  0 &  0	&  0 &  269 & 0 & 2   & 0 \\
  \hline
  \tabincell{l}{$BQ_3$} & \tabincell{l}{$\surd$} & 0 &  0 &187  & 0 &  0 &  0	&  0 & 187 & 0 & 0   & 0\\
  \cline{1-13}
   \tabincell{l}{$BQ_4$} &\tabincell{l}{$\surd$}  & 39,842	& 2,699	& 45,723	& 2,511	& 1	& 88,076		& 93	& 88,169	& 5	& 4	& 4 \\
  \cline{1-13}
   \tabincell{l}{$BQ_5$} & \tabincell{l}{$\surd$} & 45,962	& 1,929	& 6,858	& 1,504	& 1	& 54,324	& 	2	& 54,326	& 16	& 12	& 11 \\
   \cline{1-13}
   \tabincell{l}{$BQ_6$} &  & 19,663	& 1,047	& 1,589	& 756	& 1	& 22,008	& 	2	& 22,010	& 0	& 0	& 0 \\
  \cline{1-13}
    \tabincell{l}{$BQ_7$} &  & 35,849	& 3,071	& 21,233	& 2,848	& 1	& 59,930	& 	24	& 59,954	& 0	& 0	& 0 \\
  \hline
  \end{tabular}
  \end{threeparttable}
  \label{table:queriesperformanceBTC}
\end{table*}

We conduct all experiments on a cluster of 12 machines running Linux, each of which has two CPU with six cores of 1.2 GHz.
Each machine has 128 GB memory and 28 TB disk storage. We select one of these machines as the coordinator machine.
We use MPICH-3.0.4 running on C++ for communication.
By default, we use a hash function $H(v)$ to assign each vertex v in RDF graph to the i-th fragment if $H(v)\ MOD\ N = i$, where $N = 12$ is the number of machines. Each machine stores a single fragment. %By default, we use the uniform hash function and $N = 12$.

In this study, we revise gStore \cite{Zou:2013fk} to find local partial matches at each site. We denote our method as gStore$^D$.
We compare our approach with four other systems, including DREAM \cite{PVLDB2015:DREAM}, S2X \cite{DBLP:S2X}, S2RDF \cite{VLDB2016:S2RDF} and CliqueSquare \cite{ICDE2015:CliqueSquare}. The codes of these systems were released by \cite{DBLP:journals/pvldb/AbdelazizHKK17} in GitHub\footnote{https://github.com/ecrc/rdf-exp}. We also release our codes in GitHub\footnote{https://github.com/bnu05pp/gStoreD}.

\subsection{Evaluation of Each Stage}\label{sec:StageEvaluation}
In this experiment, we study the performance of our approaches at each stage (i.e., partial evaluation and assembly process)
with regard to different queries in LUBM 100M, YAGO2 and BTC. We report the running time of each stage, the size of the data shipment, the number of intermediate and complete results, and the communication time with regard to different queries in Tables \ref{table:queriesperformanceLUBM}, \ref{table:queriesperformanceYAGO2} and \ref{table:queriesperformanceBTC}.
Generally, the query performance mainly depends on two factors: the shape of the query graph and the existence of the selective triple patterns.

%\footnote{A triple pattern t is a ``selective triple pattern'' if it has no more than 100 matches in RDF graph G}
%in the query graph. the existence of the triple patterns with property variables.

For the shape of the query graph, we divide all benchmark queries into two categories according to the complexities of their structures: stars and other shapes.
The evaluation times for star queries ($LQ_2$, $LQ_4$ and $LQ_5$ in LUBM, and $BQ_1$, $BQ_2$ and $BQ_3$ in BTC) are short. Each crossing edge in the distributed RDF graph is replicated, so any results of star queries are certain to be in a single fragment, and we can directly compute out the results over each fragment without considering communications and our optimization techniques. %Thus, the evaluation times of star queries are short.
In contrast, queries of other query shapes involve multiple fragments, and generate local partial matches that increase the search space of the partial evaluation and the assembly process. % while using the optimization techniques.
Thus, evaluating them has a worse performance.

For the selective triple patterns, our method processes queries with selective triple patterns faster than queries without selective triple patterns. The performance of our method is dependent on the computation and assembly of local partial matches. The selective triple patterns can be used to filter out many irrelevant candidates and local partial matches, which significantly reduces the search space for computing and joining the local partial matches.
%Thus, if there are some selective triple patterns in the query, the performance of the query is better.

\subsection{Evaluation of Different Optimizations}
This experiment uses LUBM 100M and YAGO2 to test the effect of the three optimization techniques proposed in this study. Here, because star queries can be evaluated without involving any optimization techniques,
we only consider the benchmark
queries of other shapes ($LQ_1$, $LQ_3$, $LQ_6$ and $LQ_7$ in LUBM and all queries in YAGO2). We design a baseline that does not utilize any proposed optimization techniques (denoted as $gStore^D$-Basic), a baseline using only the optimization of the LEC feature-based assembly (denoted as $gStore^D$-LA), and a baseline using only the optimizations of the LEC feature-based assembly and LEC feature-based optimization (denoted as $gStore^D$-LO).
Fig. \ref{fig:OptimizationEvaluation} shows the experiment results.

In general, the optimization of LEC feature-based assembly only repartitions the local partial matches to reduce the join space and does not lead to extra communications, so $gStore^D$-LA has the same partial evaluation stage as $gStore^D$-Basic, and their difference is only on the assembly stage. Because $gStore^D$-LA optimizes the joining order without the extra communications, it is always faster than $gStore^D$-Basic. In contrast, the optimizations of assembling variables' internal candidates and LEC feature-based optimization lead to extra communications for internal candidates and local partial matches, so they may result in extra processing times. However, the optimizations are effective, and improve the performance in most cases. Especially for the selective queries of complex shapes ($LQ_3$ in LUBM and $YQ_1$, $YQ_2$, $YQ_4$ in YAGO2), the optimizations can improve the performance by orders of magnitude.

\begin{figure}[h]
   \centering
\subfigure[{LUBM 100M}]{%
		\resizebox{0.47\columnwidth}{!}{
				\begin{tikzpicture}[font=\Large]
 		 \begin{semilogyaxis}[
    			ybar,
        ymin=100,
ymax=100000000,
  ylabel = {Query Response Time (ms)},
   bar width=5pt,
   ymajorgrids = true,
   enlarge x limits=0.15,
    			symbolic x coords = {$LQ_1$,$LQ_3$,$LQ_6$,$LQ_7$},
    legend pos= north west,
 legend cell align=left
   		]

\addplot  coordinates {($LQ_1$, 117300.170) ($LQ_3$, 86269.018) ($LQ_6$, 3627.201) ($LQ_7$, 277564.979)};

  \addplot  coordinates {($LQ_1$, 114843.689) ($LQ_3$, 85001.779) ($LQ_6$, 3625.808) ($LQ_7$, 253570.868)};

    \addplot  coordinates {($LQ_1$, 53775.856) ($LQ_3$, 34319.261) ($LQ_6$, 3875.202) ($LQ_7$, 86802.065)};

    \addplot  coordinates {($LQ_1$, 40172.642) ($LQ_3$, 3363.339) ($LQ_6$, 3142.952) ($LQ_7$, 48511.743)};

   		  \legend{$gStore^D-Basic$ ,$gStore^D-LA$,$gStore^D-LO$,$gStore^D$}
  		\end{semilogyaxis}
\end{tikzpicture}
		}
       \label{fig:OptimizationEvaluationLUBM}%
       }
   \subfigure[{YAGO2}]{%
		\resizebox{0.47\columnwidth}{!}{
				\begin{tikzpicture}[font=\Large]
 		 \begin{semilogyaxis}[
    			ybar,
    ymin=100,
ymax=100000000,
  ylabel = {Query Response Time (ms)},
   bar width=5pt,
   ymajorgrids = true,
   enlarge x limits=0.15,
    			symbolic x coords = {$YQ_1$,$YQ_2$,$YQ_3$,$YQ_4$},
    legend pos= north west,
 legend cell align=left
   		]

\addplot  coordinates {($YQ_1$, 304954.229) ($YQ_2$, 130211.536) ($YQ_3$, 104716.408) ($YQ_4$, 204996.158)};

\addplot  coordinates {($YQ_1$, 104597.460) ($YQ_2$, 85001.779) ($YQ_3$, 28017.167) ($YQ_4$, 114564.306)};

   \addplot  coordinates {($YQ_1$, 72777.858) ($YQ_2$, 3023.756) ($YQ_3$, 28850.878) ($YQ_4$, 67249.811)};

    		\addplot  coordinates {($YQ_1$, 2153.174) ($YQ_2$, 1342.474) ($YQ_3$, 27636.616) ($YQ_4$, 4712.494)};

   		  \legend{$gStore^D-Basic$ ,$gStore^D-LA$,$gStore^D-LO$,$gStore^D$}
  		\end{semilogyaxis}
\end{tikzpicture}
		}
       \label{fig:OptimizationEvaluationYAGO2}%
       }
 \caption{ Evaluation of Different Optimizations}
 \label{fig:OptimizationEvaluation}
\end{figure}
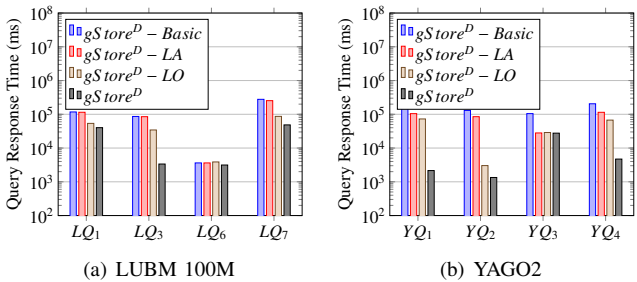

\subsection{Evaluation of Different Partitioning Strategies}\label{sec:PartitioningTest}
The aim of this experiment is to highlight the differences among different partitioning strategies. In this experiment, we use LUBM 100M and YAGO2 and test three partitioning strategies, hash partitioning, semantic hash partitioning \cite{VLDB13:SHAPE}, and METIS \cite{DBLP:metis}. Here, we also only consider the benchmark queries of other shapes. Table \ref{table:PartitioningCost} shows the costs of the different partitionings defined in Section \ref{sec:Partitioning}, while
%Fig. \ref{fig:PartitioningEvaluationSize} and
Fig. \ref{fig:PartitioningEvaluationTime} shows the evaluation times of our method over different partitionings.

The hash partitioning can uniformly distribute vertices and crossing edges among different fragments. Hence, the cost of the hash partitioning is not too high. The semantic hash partitioning is based on the URI hierarchy. For LUBM 100M, because different entities have different URI hierarchies, the semantic hash partitioning can partition the entities totally based on their domains, which greatly reduces its partitioning cost.  
In contrast, all entities in YAGO2 have the same URI hierarchy, and the cost of the semantic hash partitioning is approximately same as the hash partitioning. Hence, the performance of our method over LUBM 100M in the semantic hash partitioning is better than other partitionings, while the performance over YAGO2 is similar. In addition, although there are fewer crossing edges in METIS, its partitioning result is much more imbalanced than the hash partitioning, indicating that the cost of METIS is high. Hence, the performance in METIS is always worse than the hash partitioning for YAGO2.

\begin{table}[h]
\caption{$Cost_{Partitioning}$}
\centering
\scriptsize
\begin{threeparttable}
  \begin{tabular}{|r|r|r|r|}
  \hline
     & Hash & Semantic Hash& METIS \\
  \hline
    YAGO2 & $0.76\times 10^{14}$ & $0.77\times 10^{14}$ & $1.49\times 10^{14}$\\
  \hline
    LUBM 100M & $0.92\times 10^{9}$ & $0.55\times 10^{9}$ & $0.67\times 10^{9}$\\
  \hline
  %  BTC & $5.31\times 10^{14}$ & $4.34\times 10^{14}$ &- \\
  %\hline
  \end{tabular}
  \end{threeparttable}
  \label{table:PartitioningCost}
\end{table}

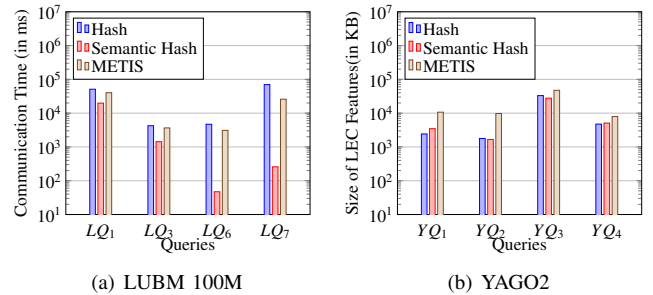
\begin{figure}[h]
   \centering
\subfigure[{LUBM 100M}]{%
		\resizebox{0.47\columnwidth}{!}{
				\begin{tikzpicture}[font=\Large]
 		 \begin{semilogyaxis}[
                anchor={(0,100)},
    			major x tick style = transparent,
    			ybar,
    			ymin = 10,
    			ymax = 10000000,
   			ymajorgrids = true,
   			ylabel = {Communication Time (in ms)},
    			xlabel = {Queries},
    			symbolic x coords = {$LQ_1$,$LQ_3$,$LQ_6$,$LQ_7$},
    			scaled y ticks = true,
			bar width=4.5pt,
             enlarge x limits=0.2,
			legend pos= north west,
 legend cell align=left
   		]

   \addplot coordinates {($LQ_1$, 51123.756) ($LQ_3$, 4243.363) ($LQ_6$, 4692.110) ($LQ_7$, 69945.621)};

   \addplot coordinates {($LQ_1$, 19809.461) ($LQ_3$, 1429.671) ($LQ_6$, 47.194) ($LQ_7$, 258) };

    		\addplot coordinates {($LQ_1$, 40300.378) ($LQ_3$, 3659.942) ($LQ_6$, 3093.450) ($LQ_7$, 25877.639)};

   		  \legend{Hash,Semantic Hash,METIS}
  		\end{semilogyaxis}
\end{tikzpicture}
		}
       \label{fig:LUBM100MPartitioningEvaluationTime}%
       }
       \subfigure[{YAGO2}]{%
		\resizebox{0.47\columnwidth}{!}{
				\begin{tikzpicture}[font=\Large]
 		 \begin{semilogyaxis}[
                anchor={(0,100)},
    			major x tick style = transparent,
    			ybar,
    			ymin = 10,
    			ymax = 10000000,
   			ymajorgrids = true,
   			ylabel = {Size of LEC Features(in KB)},
    			xlabel = {Queries},
    			symbolic x coords = {$YQ_1$,$YQ_2$,$YQ_3$,$YQ_4$},
    			scaled y ticks = true,
			bar width=4.5pt,
             enlarge x limits=0.2,
			legend pos= north west,
 legend cell align=left
   		]

\addplot coordinates {($YQ_1$, 2427.144) ($YQ_2$, 1783.421) ($YQ_3$, 33073.558) ($YQ_4$, 4748.009 ) };

  \addplot coordinates {($YQ_1$, 3465.401) ($YQ_2$, 1663.534) ($YQ_3$, 27707.272) ($YQ_4$, 5078.912) };

    		\addplot coordinates {($YQ_1$, 10687.940) ($YQ_2$, 9761.987) ($YQ_3$, 47413.566) ($YQ_4$, 7918.393) };
    		
   		  \legend{Hash,Semantic Hash,METIS}
  		\end{semilogyaxis}
\end{tikzpicture}
		}
       \label{fig:YAGO2PartitioningEvaluationTime}%
       }%
 \caption{Evaluation Time of Different Partitioning Strategies}
 \label{fig:PartitioningEvaluationTime}
\end{figure}

\subsection{Scalability Test}
We investigate the effect of data size on query evaluation times in this experiment. We generate three LUBM datasets, varying from 100 million to 1 billion triples, to test our method. 
Fig. \ref{fig:ScalabilityTest} shows the experiment results. As mentioned in Section \ref{sec:StageEvaluation}, we divide the queries into four categories according to their structures: star queries ($LQ_2$, $LQ_4$ and $LQ_5$) and other queries ($LQ_1$, $LQ_3$, $LQ_6$ and $LQ_7$).

In general, because the number of crossing edges linearly increases as the data size increases and our approach is partition bounded, the query response time also increases proportionally to the data size. Here, for queries of other shapes, the query response times may grow faster. This is because the other query graph shapes cause more complex operations in query processing, such as joining and assembly, and a larger number of local partial matches. However, even for queries of complex structures, the query performance is scalable with the RDF graph size on the benchmark datasets.

\begin{figure}[h]
   \centering
   \subfigure[{Star Queries}]{%
		\resizebox{0.47\columnwidth}{!}{
			\begin{tikzpicture}[font=\large]
    \begin{semilogyaxis}[
        xlabel=Datasets,
        ylabel=Query Response Time (in ms),
        xtick = {10,50,100},
        xticklabels={100M,500M,1B},
        legend cell align=left,
        legend style={draw=none},
        legend pos= north west
    ]
      \addplot plot[mark size=3.5pt] coordinates {
        (10,     8488)
        (50,   33679 )
        (100,    60254.06245)
    };
    \addplot plot[mark size=3.5pt]  coordinates {
        (10,     221)
        (50,    399.9851038)
        (100,    505.2395287)
    };
    \addplot plot[mark=triangle*,mark size=3.5pt]  coordinates {
        (10,    187)
        (50,    332.8495377)
        (100,   631.6011569)
    };

    \legend{$LQ_2$\\$LQ_4$\\$LQ_5$\\}

    \end{semilogyaxis}
\end{tikzpicture}
		}
       \label{fig:LinearScalabilityTest}%
       }%
\subfigure[{Other Queries}]{%
		\resizebox{0.47\columnwidth}{!}{
			\begin{tikzpicture}[font=\large]
    \begin{semilogyaxis}[
        xlabel=Datasets,
        ylabel=Query Response Time (in ms),
        xtick = {10,50,100},
        xticklabels={100M,500M,1B},
        legend cell align=left,
        legend style={draw=none},
        legend pos= north west
    ]
      \addplot plot[mark size=3.5pt] coordinates {
        (10,     40172)
        (50,   140430.0109 )
        (100,    316447.014)
    };
    \addplot plot[mark size=3.5pt]  coordinates {
        (10,    3363)
        (50,    15277.73319)
        (100,    30108.07983)
    };
    \addplot plot[mark=triangle*,mark size=3.5pt]  coordinates {
        (10,    3142)
        (50,    12594.30993)
        (100,   23103.9144)
    };
     \addplot plot[mark size=3.5pt]  coordinates {
        (10,     48511)
        (50,    167809.1843)
        (100,   289493.375)
    };

    \legend{$LQ_1$\\$LQ_3$\\$LQ_{6}$\\$LQ_7$\\}

    \end{semilogyaxis}
\end{tikzpicture}
		}
       \label{fig:SnowflakeScalabilityTest}%
       }
 \caption{Scalability Test}
 \label{fig:ScalabilityTest}
\end{figure}
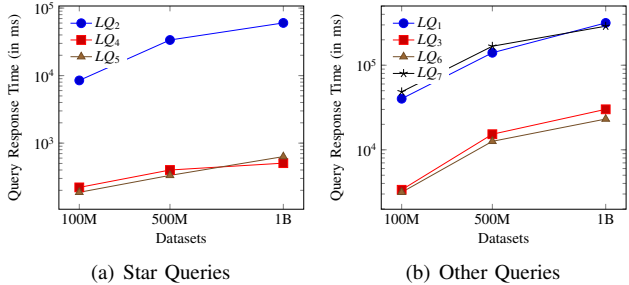

\subsection{Online Performance Comparison}
In this experiment, we evaluate the online performance of our method on different partitionings of three datasets, YAGO2, LUBM 1B, and BTC. Fig. \ref{fig:OnlineComparison} shows the performances. Note that METIS can only be used on YAGO2, and fails to partition LUBM 1B and BTC in our setting.%Note that, METIS can only be used to YAGO2, and fails to partition LUBM 1B and BTC in our setting.

The results of this experiment include a comparative evaluation of our method against four state-of-the-art public disk-based distributed RDF systems proposed in the most recent three years, including DREAM \cite{PVLDB2015:DREAM}, S2X \cite{DBLP:S2X}, S2RDF \cite{VLDB2016:S2RDF}, and CliqueSquare \cite{ICDE2015:CliqueSquare}, which are provided by \cite{DBLP:journals/pvldb/AbdelazizHKK17}. Other distributed RDF systems in the most recent three years are either unreleased, or are memory-based systems that are in different environments than targeted in this study. Note that S2X fails to run all queries on LUBM 1B. We also run DREAM and CliqueSquare over BTC, while S2X and S2RDF fail over BTC.

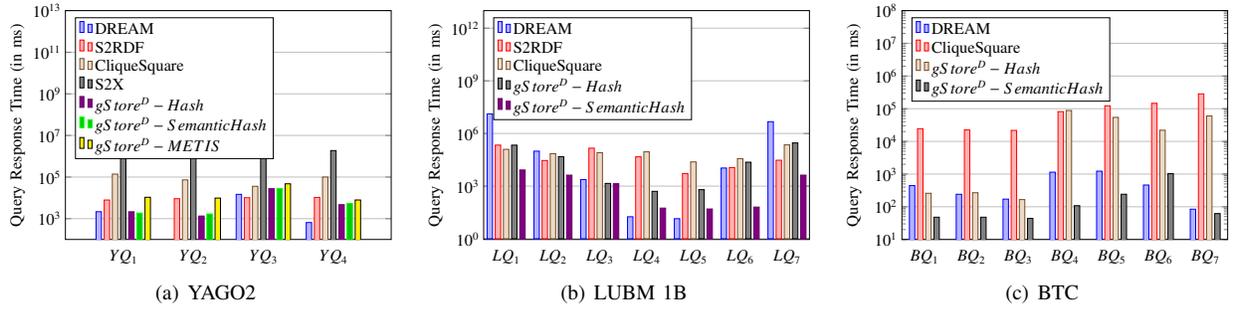
\begin{figure*}
   \centering
   \subfigure[{YAGO2}]{%
	\resizebox{0.6\columnwidth}{!}{
			\begin{tikzpicture}[font=\large]
 		 \begin{semilogyaxis}[
                anchor={(0,100)},
    			major x tick style = transparent,
    			ybar,
         width = 10cm,
               height = 7.5cm,
    ymax = 10000000000000,
    ymin = 100,
   			ymajorgrids = true,
   			ylabel = {Query Response Time (in ms)},
    			symbolic x coords = {$YQ_1$,$YQ_2$,$YQ_3$,$YQ_4$},
    			scaled y ticks = true,
			bar width=4pt,
             enlarge x limits=0.275,
			legend pos= north west,
 legend cell align=left
   		]

   \addplot  coordinates {($YQ_1$, 2161) ($YQ_2$, 0) ($YQ_3$, 14751) ($YQ_4$, 651) };

    		\addplot coordinates {($YQ_1$, 7893) ($YQ_2$, 9128) ($YQ_3$, 10161) ($YQ_4$, 10473) };

    \addplot  coordinates {($YQ_1$, 139021) ($YQ_2$, 73011) ($YQ_3$, 36006) ($YQ_4$, 100015) };

    		\addplot coordinates {($YQ_1$, 1811658) ($YQ_2$, 1863374) ($YQ_3$, 1720424) ($YQ_4$, 1876964) };

   \addplot coordinates {($YQ_1$, 2153) ($YQ_2$, 1340) ($YQ_3$, 27636) ($YQ_4$, 4712)};

   \addplot coordinates {($YQ_1$, 1870) ($YQ_2$, 1688.261) ($YQ_3$,  28300.534) ($YQ_4$, 5364.588)};

    		\addplot [fill=yellow] coordinates {($YQ_1$, 10687.940) ($YQ_2$, 9761.987) ($YQ_3$, 47413.566) ($YQ_4$, 7918.393)};

   		 \legend{DREAM, S2RDF, CliqueSquare, S2X, $gStore^D-Hash$, $gStore^D-SemanticHash$, $gStore^D-METIS$}
  		\end{semilogyaxis}
\end{tikzpicture}
	}
 \label{fig:yago2Comparison}%
}
 \subfigure[{LUBM 1B}]{%
		\resizebox{0.6\columnwidth}{!}{
				\begin{tikzpicture}[font=\large]
 		 \begin{semilogyaxis}[
                anchor={(0,100)},
    			major x tick style = transparent,
    			ybar,
         width = 10cm,
               height = 7.5cm,
    ymax = 10000000000000,
   			ymajorgrids = true,
   			ylabel = {Query Response Time (in ms)},
    			symbolic x coords = {$LQ_1$,$LQ_2$,$LQ_3$,$LQ_4$,$LQ_5$,$LQ_6$,$LQ_7$},
    			scaled y ticks = true,
			bar width=4pt,
             enlarge x limits=0.08,
			legend pos= north west,
 legend cell align=left
   		]

\addplot coordinates {($LQ_1$, 13031410) ($LQ_2$, 98263) ($LQ_3$, 2358) ($LQ_4$, 18) ($LQ_5$, 14) ($LQ_6$, 10755) ($LQ_7$, 4700381) };

\addplot coordinates {($LQ_1$, 217537) ($LQ_2$, 28917) ($LQ_3$, 145761) ($LQ_4$, 46770) ($LQ_5$, 5233) ($LQ_6$, 11308) ($LQ_7$, 29965)};

   \addplot  coordinates {($LQ_1$, 125020) ($LQ_2$, 71010) ($LQ_3$, 80010) ($LQ_4$, 90010) ($LQ_5$, 24000) ($LQ_6$, 37010) ($LQ_7$, 224040)};

   \addplot coordinates {($LQ_1$, 216447.014) ($LQ_2$, 47662.64815) ($LQ_3$, 1429.671) ($LQ_4$, 505.2395287) ($LQ_5$, 631.6011569) ($LQ_6$, 23103.9144) ($LQ_7$, 289493.375) };

\addplot coordinates {($LQ_1$, 8344) ($LQ_2$, 4235) ($LQ_3$, 1377) ($LQ_4$, 56) ($LQ_5$, 50) ($LQ_6$, 64) ($LQ_7$, 4239) };

   		 \legend{DREAM, S2RDF, CliqueSquare, $gStore^D-Hash$, $gStore^D-SemanticHash$}
  		\end{semilogyaxis}
\end{tikzpicture}
		}
 \label{fig:LUBM1BComparison}%
  }
  \subfigure[{BTC}]{%
		\resizebox{0.6\columnwidth}{!}{
				\begin{tikzpicture}[font=\large]
 		 \begin{semilogyaxis}[
                anchor={(0,100)},
    			major x tick style = transparent,
    			ybar,
         width = 10cm,
               height = 7.5cm,
    ymax = 100000000,
    ymin = 10,
   			ymajorgrids = true,
   			ylabel = {Query Response Time (in ms)},
    			symbolic x coords = {$BQ_1$,$BQ_2$,$BQ_3$,$BQ_4$,$BQ_5$,$BQ_6$,$BQ_7$},
    			scaled y ticks = true,
			bar width=4pt,
             enlarge x limits=0.08,
			legend pos= north west,
 legend cell align=left
   		]

   \addplot coordinates {($BQ_1$, 444.567) ($BQ_2$, 241.317) ($BQ_3$, 171.516) ($BQ_4$, 1143.66) ($BQ_5$, 1228.86) ($BQ_6$, 462.784) ($BQ_7$, 84.4228) };

   \addplot coordinates {($BQ_1$, 24445.67) ($BQ_2$, 22413.17) ($BQ_3$, 21715.16) ($BQ_4$, 81436.6) ($BQ_5$, 122886) ($BQ_6$, 146278.4) ($BQ_7$, 284422.8) };

   \addplot coordinates {($BQ_1$, 259) ($BQ_2$, 269) ($BQ_3$, 166) ($BQ_4$, 88169) ($BQ_5$, 54326) ($BQ_6$, 22010) ($BQ_7$, 59954) };

\addplot coordinates {($BQ_1$, 48) ($BQ_2$, 48) ($BQ_3$, 44) ($BQ_4$, 108) ($BQ_5$, 242) ($BQ_6$, 1029) ($BQ_7$, 62) };

   		 \legend{DREAM, CliqueSquare,$gStore^D-Hash$,$gStore^D-SemanticHash$}
  		\end{semilogyaxis}
\end{tikzpicture}
		}
 \label{fig:BTCComparison}%
  }
 \caption{Online Performance Comparison}
 \label{fig:OnlineComparison}
\end{figure*}

Generally, our method is partitioning-tolerant, and the performances of our method over different partitionings show the superiority of our proposed approach.

In particular, S2X, S2RDF, and CliqueSquare are three cloud-based systems that suffer from the expensive overhead of scans and joins in the cloud. Only when the queries ($LQ_1$, $LQ_2$ and $LQ_7$ in LUBM) are unselective and are evaluated over a very large RDF dataset (LUBM 1B) that can generate many intermediate results might they have better performances than DREAM and our approach when running over ill-suited partitionings. However, when our method runs over partitionings with the smallest costs (hash partitioning for YAGO2 and semantic hash partitioning for LUBM 1B and BTC), our method can outperform others.

On the other hand, when the queries ($LQ_3$, $LQ_4$, $LQ_5$ and $LQ_6$ in LUBM 1B and all queries in BTC) are selective or the RDF dataset (YAGO2) is not very large, DREAM \cite{PVLDB2015:DREAM} and our system can outperform the cloud-based systems in most cases. Here, DREAM builds a single RDF-3X database for the entire dataset in each site, and decomposes the input query into multiple star-shape subqueries, where each subquery is answered by a single site. This can greatly reduce the performances over the selective queries and small datasets. However, DREAM exhibits excessive replication, and causes huge overhead when processing complex queries. When a query is complex, it may lead to multiple large subqueries. Evaluating the large
subqueries over a site of the entire dataset often results in many intermediate results, and joining these intermediate
results is also costly. Our method, running over the best partitionings, can always be comparable to DREAM. In addition, DREAM fails to process $YQ_2$.

\section{Related Work}\label{sec:relatedwork}
\textbf{Distributed SPARQL Query Processing.}
There have been many works on distributed SPARQL query processing, and a very good survey is \cite{VLDBJ15:RDFCloudSurvey}.
 In recent years, some approaches such as \cite{ICDE2015:CliqueSquare,TKDE2015:DiploCloud,DBLP:conf/icde/WuZYLJ15,PVLDB2015:AdHash,VLDB2016:AdPart,DBLP:journals/corr/Peng0ZZ15,PVLDB2015:DREAM,VLDB2016:S2RDF,PVLDB2017:Stylus} have been proposed. We classify them into three classes: cloud-based approaches, partitioning based approaches, and partitioning-tolerant approaches.

First, some recent works (e.g., \cite{ICDE2015:CliqueSquare,VLDB2016:S2RDF,DBLP:S2X}) focus on managing RDF datasets using cloud platforms. CliqueSquare \cite{ICDE2015:CliqueSquare} discusses how to build query plans by relying on n-ary (star) equality joins in Hadoop. S2RDF \cite{VLDB2016:S2RDF} uses Spark SQL to store the RDF data in a vertical partitioning schema and materializes some extra join relations. In the online phase, S2RDF transforms the query into SQL queries
and merges the results of the SQL queries.
S2X \cite{DBLP:S2X} uses GraphX in Spark to evaluate SPARQL queries. S2X first distributes all triple patterns to all vertices. Then, vertices validate their triple candidacy with their neighbors by exchanging messages. Lastly, the partial results are collected and merged. Stylus \cite{PVLDB2017:Stylus} uses Trinity \cite{SIGMOD2013:Trinity}, a distributed in-memory key-value
store, to maintain the adjacent list of the RDF graph while considering the types. In the online phase, Stylus decomposes the query into multiple star subqueries and evaluates the subqueries by using the interfaces of Trinity.

Second, some approaches \cite{TKDE2015:DiploCloud,DBLP:conf/icde/WuZYLJ15,PVLDB2015:AdHash,VLDB2016:AdPart,DBLP:journals/corr/Peng0ZZ15} are partition-based. They divide an RDF graph into several partitions. Each partition is placed at a site that installs a centralized RDF system to manage it. At run time, a SPARQL query is decomposed into several subqueries that can be answered locally at a site. The results of the subqueries are finally merged. Each of these approaches has its own data partitioning strategy, and different partitioning strategies result in different query decomposition methods. 
DiploCloud \cite{TKDE2015:DiploCloud} asks the administrator to define some templates as the partition unit. DiploCloud stores the instantiations of the templates in compact lists as in a column-oriented database system; PathBMC \cite{DBLP:conf/icde/WuZYLJ15} adopts the end-to-end path as the partition unit to partition the data and query graph; AdHash \cite{PVLDB2015:AdHash} and AdPart \cite{VLDB2016:AdPart} directly use the subject values to partition the RDF graph and mainly discuss how to reduce the communication cost during distributed query evaluation; and Peng et al. \cite{DBLP:journals/corr/Peng0ZZ15} mine some frequent patterns in the log as the partitioning units.

DREAM \cite{PVLDB2015:DREAM} and Peng et al. \cite{DBLP:journals/corr/PengZOCZ14} are two other approaches that neither partition RDF graphs nor use existing cloud platforms. In DREAM \cite{PVLDB2015:DREAM}, each site maintains the whole RDF dataset. For query processing, DREAM divides the input query into subqueries, and executes each subquery in a site. The intermediate results are merged to produce the final matches. Peng et al. \cite{DBLP:journals/corr/PengZOCZ14} propose a partition-tolerant distributed approach based on the ``partial evaluation and assembly'' framework. However, its efficiency has a large potential for improvement.

\textbf{Partial Evaluation}. As surveyed in \cite{DBLP:journals/csur/Jones96}, partial evaluation has found many applications ranging from compiler optimization to distributed evaluation of functional programming languages.
Recently, partial evaluation has been used for evaluating queries on distributed graphs, as in \cite{DBLP:conf/www/MaCHW12,DBLP:journals/pvldb/FanWW12,DBLP:journals/pvldb/FanWWD14,DBLP:conf/sigmod/GurajadaT16,Wang:2016:EDR:2983323.2983877}. In \cite{DBLP:journals/pvldb/FanWW12,DBLP:conf/sigmod/GurajadaT16}, the authors provide algorithms for evaluating reachability queries on distributed graphs based on partial evaluation. In \cite{DBLP:conf/www/MaCHW12,DBLP:journals/pvldb/FanWWD14}, the authors study partial evaluation algorithms and
optimizations for distributed graph simulation. Wang et al. \cite{Wang:2016:EDR:2983323.2983877} discuss how to answer regular path queries on large-scale RDF graphs using partial evaluation.
Peng et al. \cite{DBLP:journals/corr/PengZOCZ14} discuss how to employ the ``partial evaluation and assembly'' framework to handle SPARQL queries, but it fails to provide performance guarantees on the total network traffic and the response time.
\section{Conclusion}\label{sec:Conclusion}
In this study, we propose three optimizations to improve the partial evaluation-based distributed SPARQL query processing approach. The first is to compress the partial evaluation results in a compact data structure named the \emph{LEC feature}, and to communicate them among sites to filter out some irrelevant partial evaluation results while providing some performance guarantees. The second is the LEC feature-based assembly of all local partial matches to reduce the search space. Moreover, we propose an optimization that communicates variables' candidates among the sites to avoid irrelevant local partial matches. We also discuss the impact of different partitionings over our approach. In addition, we perform extensive experiments to confirm our approach.

\small{
\textbf{Acknowledgement.} This work was supported by The National Key Research and Development Program of China under grant 2018YFB1003504, NSFC under grant 61702171, 61622201 and 61532010, Hunan Provincial Natural Science Foundation of China under grant 2018JJ3065, and the Fundamental Research Funds for the Central Universities.
}

\bibliographystyle{abbrv}
\bibliography{sigproc}

\end{document}